\documentclass[aps,onecolumn,superscriptaddress,dvipsnames,nofootinbib,amssymb,floatfix,11pt,prx]{revtex4-2} 
\usepackage{graphicx}
\usepackage{mathpazo}
\usepackage{dcolumn}
\usepackage{bm}
\usepackage{qcircuit}
\usepackage{braket}
\usepackage{color}
\usepackage{amsthm}
\usepackage{amsmath}
\usepackage{chemarrow}
\usepackage{overpic}
\usepackage{diagbox}
\usepackage{braket}
\usepackage{subfigure}
\usepackage{makecell}
\newtheorem{theorem}{Theorem}
\newtheorem{definition}{Definition}
\newtheorem{lemma}{Lemma}
\newtheorem{corollary}{Corollary}
\usepackage{chngcntr}
\usepackage[colorlinks=true,linkcolor=red,bookmarks=true,breaklinks=true,citecolor=blue]{hyperref}
\let\oldacl\addcontentsline
\renewcommand{\addcontentsline}[3]{}
\newcommand{\Tr}{\mathrm{Tr}}
\newcommand{\re}{\text{Re}}
\newcommand{\im}{\text{Im}}

\newcommand{\polylog}{\text{polylog}}

\definecolor{myblue}{RGB}{0,0,255}
\usepackage{setspace}

\allowdisplaybreaks

\begin{document}
\title{Exponential Lindbladian fast forwarding and exponential amplification of certain Gibbs state properties}

\author{Zhong-Xia Shang}
\email{shangzx@hku.hk}
\affiliation{HK Institute of Quantum Science $\&$ Technology, The University of Hong Kong, Hong Kong}
\affiliation{QICI Quantum Information and Computation Initiative, School of Computing and Data Science,
The University of Hong Kong, Pokfulam Road, Hong Kong}

\author{Dong An}
\email{dongan@pku.edu.cn} 
\affiliation{Beijing International Center for Mathematical Research, Peking University, Beijing, China, 100871}

\author{Changpeng Shao}
\email{changpeng.shao@amss.ac.cn} 
\affiliation{SKLMS, Academy of Mathematics and Systems Science, Chinese Academy of Sciences, Beijing, China, 100190}

\begin{abstract}

We investigate Lindbladian fast-forwarding and its applications to estimating Gibbs state properties. Fast-forwarding refers to the ability to simulate a system of time $t$ using significantly fewer than $t$ queries or circuit depth. While various Hamiltonian systems are known to circumvent the no fast-forwarding theorem, analogous results for dissipative dynamics, governed by Lindbladians, remain largely unexplored. We first present a quantum algorithm for simulating purely dissipative Lindbladians with unitary jump operators, achieving additive query complexity $
\mathcal{O}\left(t + \log(\varepsilon^{-1})\right)$ up to error~$\varepsilon$, improving previous algorithms. When the jump operators have certain structures (i.e., block-diagonal Paulis), the algorithm can be modified to achieve exponential fast-forwarding, attaining circuit depth $\mathcal{O}\left(\log\left(t + \log(\varepsilon^{-1})\right)\right)$, while preserving query complexity via parallel access.
Using these fast-forwarding techniques, we develop a quantum algorithm for estimating Gibbs state properties of the form
$\langle \psi_1 | e^{-\beta(H + I)} | \psi_2 \rangle$, up to additive error $\epsilon$, with $H$ the Hamiltonian and $\beta$ the inverse temperature.
For input states exhibiting certain coherence conditions ---e.g.,~$\langle 0|^{\otimes n} e^{-\beta(H + I)} |+\rangle^{\otimes n}$---our method achieves exponential improvement in complexity (measured by circuit depth),
$\mathcal{O} (2^{-n/2} \epsilon^{-1} \log \beta  ),$
compared to the quantum singular value transformation-based approach, with complexity
$\tilde{\mathcal{O}} (\epsilon^{-1} \sqrt{\beta} )$. We show how to apply this exponential improvement to applications such as the ground state overlap testing and amplitude estimation. For general $| \psi_1 \rangle$ and $| \psi_2 \rangle$, we also show how the level of improvement is changed with the coherence resource in $| \psi_1 \rangle$ and $| \psi_2 \rangle$. 

\end{abstract}

\maketitle

\section{Introduction}

Quantum computing has emerged as a revolutionary computational paradigm, offering the potential to solve problems that are intractable for classical computers~\cite{shor1994algorithms,grover1996fast,feynman2018simulating}. Among its most transformative applications is the simulation of quantum systems~\cite{feynman2018simulating}. The development of quantum simulation algorithms follows a central theme of reducing computational complexity. 
In the context of Hamiltonian simulation, where the goal is to simulate $e^{-iHt}$ up to error $\varepsilon$, techniques have evolved from product formulas~\cite{lloyd1996universal,childs2021theory} and the linear combination of unitaries (LCU)~\cite{berry2015simulating} to quantum signal processing (QSP)~\cite{low2019hamiltonian,low2017optimal}, achieving optimal additive query complexity $\mathcal{O}(t + \frac{\log(\varepsilon^{-1})}{\log\log(\varepsilon^{-1})}),$
which matches the known theoretical lower bound~\cite{berry2007efficient,berry2014exponential}.
The lower bound of $\Omega(t)$ arises from the no fast-forwarding theorem for general Hamiltonian simulation~\cite{berry2007efficient,atia2017fast}, which states that no quantum algorithm can simulate $e^{-iHt}$ for arbitrary Hamiltonians $H$ with query complexity sublinear in $t$. Furthermore, Ref.~\cite{chia2023impossibility} rules out the possibility of fast-forwarding even in terms of circuit depth through parallel queries.
However, these results do not exclude the possibility of fast-forwarding in special cases. In fact, the exponential speedup in Shor’s factoring algorithm can be understood as an example of exponential fast-forwarding~\cite{atia2017fast}. Other notable examples include efficiently diagonalizable Hamiltonians, commuting local Hamiltonians, and quadratic fermionic Hamiltonians~\cite{atia2017fast,gu2021fast}.

Hamiltonian evolution models the dynamics of closed quantum systems. When a system interacts with a large environment, a widely used model for describing the dynamics of open quantum systems is the Lindbladian~\cite{lindblad1976generators,gorini1976completely}. In recent years, there has been rapid growth in quantum computing research focused on simulating Lindbladian dynamics, which has various applications, including Gibbs state preparation~\cite{chen2023quantum,chen2023efficient,ding2025efficient,rouze2024optimal}, ground state preparation of quantum Hamiltonians~\cite{ding2024single,lin2025dissipative,zhan2025rapid,shang2024polynomial}, dissipative quantum state engineering~\cite{guo2025designing}, quantum optimization~\cite{chen2024local,basso2024optimizing}, and solving differential equations~\cite{shang2024design}. 

Various Lindbladian simulation algorithms have been proposed~\cite{kliesch2011dissipative,childs2016efficient,schlimgen2021quantum,pocrnic2023quantum,ding2024simulating,yu2024exponentially}, with the current state-of-the-art algorithms~\cite{cleve2016efficient,li2022simulating} achieving a multiplicative query complexity of
$\tilde{\mathcal{O}}(t\,\polylog(\varepsilon^{-1}))$. Since Hamiltonian dynamics is a special case of Lindbladian dynamics, any universal Lindbladian simulation algorithm must obey the same complexity lower bound as in Hamiltonian simulation, namely 
$\Omega(t + \frac{\log(\varepsilon^{-1})}{\log\log(\varepsilon^{-1})}),$
which rules out fast forwarding in the general case (i.e. simulating general Lindbladians of an evolution time $t$ using sublinear cost in $t$). In fact, even for purely dissipative Lindbladians (i.e., with no internal Hamiltonian), fast forwarding is excluded by a dissipative version of the no fast-forwarding theorem~\cite{childs2016efficient}. 
Nevertheless, exploring special cases that admit Lindbladian fast forwarding remains an important direction, as it may improve the efficiency of Lindbladian simulation in practically relevant cases, and deepen our understanding of the similarities and differences between coherent and dissipative quantum dynamics.

Beyond quantum simulation, a natural question is whether quantum computing can offer advantages in fundamental tasks such as understanding thermal equilibrium, and preparing the Gibbs state $\rho_\beta \propto e^{-\beta H}$, where $\beta \geq 0$ is the inverse temperature. Gibbs states characterize the static properties of quantum systems, and understanding them is not only essential for studying phases of matter in condensed matter physics~\cite{alhambra2023quantum}, but also has significant implications for practical applications in machine learning~\cite{amin2018quantum} and optimization~\cite{brandao2017quantum,van2017quantum}.

However, establishing a quantum advantage in Gibbs state preparation is highly nontrivial. At high temperatures, Gibbs states are close to the maximally mixed state and exhibit no entanglement~\cite{bakshi2024high}, making them efficiently simulable on classical computers~\cite{yin2023polynomial}. In contrast, at moderately low temperatures (e.g., $\beta = \Omega(\log n)$~\cite{rajakumar2024gibbs}), Gibbs states can already have large overlap with ground states, rendering their preparation QMA-hard~\cite{kempe2004complexity}.
Only recently have exponential quantum advantages been established in this area. These include Gibbs sampling tasks whose classical intractability is connected to the collapse of the polynomial hierarchy~\cite{rajakumar2024gibbs,bergamaschi2024quantum}, as well as results that relate Gibbs state preparation to universal quantum computation~\cite{rouze2024efficient}. Exploring further advantages and applications of quantum algorithms for Gibbs-related tasks remains an active and promising research direction.

In this work, we investigate the possibility of Lindbladian fast-forwarding and explore its applications to estimating properties of Gibbs states. Our main contributions are summarized as follows:

In the first part of this work, we explore the quantum simulation of purely dissipative Lindbladians with unitary jump operators. We first present a quantum algorithm for simulating such Lindbladians with additive query complexity 
\begin{equation}
    \mathcal{O}\left(t + \frac{\log(\varepsilon^{-1})}{\log\log(\varepsilon^{-1})} \right),
\end{equation}
as stated in Theorem~\ref{mainthe1}. This improves upon the previously known multiplicative complexity $\mathcal{O}(t~\polylog(\varepsilon^{-1}))$ from \cite{cleve2016efficient,li2022simulating}. Our approach is based on implementing the linear combination of unitaries (LCU) technique \cite{childs2012hamiltonian} in the vectorization space, where a density matrix $\rho = \sum_{ij} \rho_{ij} |i\rangle\langle j|$ is treated as a vector $\vec{\rho} = \sum_{ij} \rho_{ij} |i\rangle |j\rangle$. Thanks to the dissipative nature of the Lindbladian, there is no need for the oblivious amplitude amplification procedure typically required in LCU-based Hamiltonian simulation \cite{berry2015simulating}. 

Next, we show that when the jump operators are block-diagonal Pauli operators, this simulation algorithm can be adapted to enable exponentially fast forwarding — specifically, a circuit depth of 
\begin{equation}
    \mathcal{O}\left(\log\left(t + \frac{\log(\varepsilon^{-1})}{\log\log(\varepsilon^{-1})}\right)\right),
\end{equation}
while maintaining the same query complexity as before (see Theorem~\ref{mainthe2}). So our fast-forwarding results are in terms of depth, not in terms of query complexity. The query complexity does not change in our case. Namely, it is a depth-level fast-forwarding.
This improvement is enabled by the ability to perform parallel Pauli matrix multiplication coherently on quantum computers via the tableau representation, which allows the Lindbladian queries to be executed in parallel. While parallel queries in Hamiltonian simulation have been known to yield mild improvements in special cases \cite{zhang2024parallel,haah2021quantum}, to our knowledge, this work demonstrates for the first time an exponential separation between circuit depth and query complexity in Lindbladian simulation.

In the second part of this work, we apply this exponentially fast-forwardable class of Lindbladians to propose a quantum algorithm (see Theorem~\ref{mainthe3}) for estimating Gibbs state properties of the form $\langle \psi_1|e^{-\beta (H+I)}|\psi_2\rangle$, which we refer to as the Gibbs coherence amplitude (GCA). GCA itself is related to the physically-interested property called the complex-time Loschmidt echo (or the boundary partition function)~\cite{heyl2013dynamical,heyl2018dynamical,wei2025riemann,leclair1995boundary,santilli2020phase}. When $\beta$ goes to infinity and zero, GCA also relates to amplitude estimation~\cite{brassard2000quantum} and ground state overlap testing~\cite{zanardi2006ground}, respectively. We will show that our algorithm leads to exponential improvements in all these applications.

For a general GCA quantum estimator, a lower bound of $\Omega(\epsilon^{-1}+\sqrt{\beta})$ can be obtained where $\epsilon$ is the additive estimation error. The lower bound of $\epsilon$ follows from the lower bound of amplitude estimation at $\beta=0$~\cite{brassard2000quantum,bennett1997strengths}, and the lower bound of $\beta$ follows from the lower bound of the degree of the polynomial approximation to the $e^{-\beta x}$ function~\cite[Theorem 1.5]{montanaro2024quantum}. Previously, the popular quantum singular value transformation (QSVT) \cite{gilyen2019quantum} combined with the amplitude estimation techniques~\cite{brassard2000quantum,aaronson2020quantum,grinko2021iterative,rall2023amplitude} achieves a near-optimal complexity $\tilde{\mathcal{O}}(\epsilon^{-1} \sqrt{\beta})$ ($\tilde{\mathcal{O}}$ hides $\log(\epsilon^{-1})$ factor) which suggests little space for general further improvements. However, our motivation is that if we only care about special cases, whether the possibility of significant improvements could exist. We believe focusing on special cases of hard problems will be one of the major directions for future quantum algorithm design. Answering this question will help us to understand what properties of Gibbs states are more accessible to quantum algorithms, which is crucial for identifying impactful quantum applications and rigorously establishing quantum advantages. Here, we give an affirmative answer to the question. For certain GCAs, such as $|\psi_1\rangle = |0\rangle^{\otimes n}$ and $|\psi_2\rangle = |+\rangle^{\otimes n}$, our algorithm has the circuit depth
\begin{equation} \label{cost3}
    \mathcal{O}\left(2^{-n/2} \epsilon^{-1} \log \beta \right),
\end{equation}
achieving exponential improvements in both $\beta$ and $n$ compared with QSVT, when the spectral norm is unknown. For general choices of $\ket{\psi_1}$ and $\ket{\psi_2}$, we also provide a complexity result, which turns out highly depends on the quantum coherence resources \cite{streltsov2017colloquium}—namely, the number of Hadamard gates—required to prepare $\ket{\psi_1}$ and $\ket{\psi_2}$.

In contrast to previous algorithms \cite{chen2023quantum,chen2023efficient,ding2025efficient,rouze2024optimal} based on simulating Lindbladians with the Gibbs state as a static steady state, our approach is dynamical. We leverage the non-diagonal matrix encoding (NDME) and channel block encoding (CBE) techniques introduced in \cite{shang2024estimating,shang2024design}, treating vectorized off-diagonal density matrix blocks as pure states. Specifically, we apply a Pauli encoding, mapping $I \mapsto |0\rangle$ and $X \mapsto |1\rangle$. Assuming knowledge of the Pauli decomposition $H = \sum_i \lambda_i Q_i$ (which holds for many Hamiltonians in physics and chemistry \cite{coleman2015introduction,mcardle2020quantum}), we construct tailored Lindbladians which is within the regime of our exponential fast-forwarding, such that their dissipative evolution effectively implements $e^{-\beta(H+I)}$ on these Pauli-encoded states. The benefit of the Pauli encoding is that the $2^{n/2}$ exponentially large vector norm of the vectorized Pauli matrices can help exponentially amplify certain GCAs, achieving the complexity in Eq.~\ref{cost3}.

\section{Main results}
\subsection{Lindbladian fast forwarding}

The Lindblad master equation or Lindbladian, describes the dynamics of open quantum systems when the environment is large enough to enable the Markovian approximation. Lindbladian dynamics is a completely positive trace-preserving (CPTP) map, and is given in its general form \cite{lindblad1976generators,gorini1976completely}
\begin{equation}\label{mainlme}
\frac{d\rho(t)}{dt}=\mathcal{L}[\rho(t)]=-i[H_s,\rho(t)]+
\sum_{i=1}^M g_i\left(F_i\rho(t) F_i^\dag-
\frac{1}{2}\{\rho(t),F_i^\dag F_i\}\right),
\end{equation}
where $\rho(t)=\sum_{ij}\rho_{ij}(t)|i\rangle\langle j|$ is the density matrix of the system, $H_s$ is the system internal Hamiltonian, and $F_i$'s are quantum jump operators with jump strength $g_i\geq0$. In the above, $[A,B] := AB - BA$ denotes the commutator, and $\{A,B\} := AB + BA$ denotes the anticommutator for any operators $A$ and $B$.

In this work, we consider the simulation of purely dissipative Lindbladians with $H_s=0$ and $F_i$ as unitary operators, i.e., 
\begin{equation}\label{dlme}
\frac{d\rho(t)}{dt}=\mathcal{L}_d[\rho(t)]=\sum_{i=1}^M g_i\left(F_i\rho(t) F_i^\dag-\rho(t)\right).
\end{equation}

We define the Lindbladian norm
\begin{equation}
\|\mathcal{L}_d\|_L=\sum_{i=1}^M g_i,
\end{equation}
which matches the norm definitions in the previous Lindbladian simulation algorithms~\cite{cleve2016efficient,li2022simulating}.

Given access to the following two unitary operators
\begin{eqnarray}
&&U_g:~U_g|0\rangle=\sum_{i=1}^M\sqrt{\frac{g_i}{\|\mathcal{L}_d\|_L}}|i\rangle,\label{mainug}\\
&&U_F:~U_{F}=|0\rangle\langle 0|\otimes I+\sum_{i=1}^M |i\rangle\langle i|\otimes F_i,\label{uf}
\end{eqnarray}
we have Theorem \ref{mainthe1} for the purely Lindbladian simulation. The proof is deferred to Appendix \ref{sec: appendix B}.

\begin{definition}[Diamond distance \cite{aharonov1998quantum}]\label{diamond}
The diamond norm between two $n$-qubit quantum channels $\mathcal{C}_1$ and $\mathcal{C}_2$ is defined as
\begin{eqnarray}
\|\mathcal{C}_1 - \mathcal{C}_2\|_\diamond = \sup_\sigma \left\| \left(\mathcal{C}_1 \otimes \mathcal{I}\right)[\sigma] - \left(\mathcal{C}_2 \otimes \mathcal{I}\right)[\sigma] \right\|_1,
\end{eqnarray}
where the supremum is taken over all $2n$-qubit density matrices $\sigma$, and $\|\cdot\|_1$ denotes the trace norm.
\end{definition}

\begin{theorem}\label{mainthe1}
There is a quantum algorithm that outputs an operator $\mathcal{A}$ such that 
$\|e^{\mathcal{L}_dt}-\mathcal{A}\|_\diamond\leq \varepsilon$, with
\begin{eqnarray}
\text{queries on $U_g$ and $U_F$:} \quad && \mathcal{O}\left(\|\mathcal{L}_d\|_Lt+\frac{\log(\varepsilon^{-1})}{\log\log(\varepsilon^{-1})}\right). \\
\text{circuit depth:} \quad && \mathcal{O}\left(\left(\|\mathcal{L}_d\|_Lt+\frac{\log(\varepsilon^{-1})}{\log\log(\varepsilon^{-1})}\right)M\right).\\
\text{ancilla qubits:} \quad && \mathcal{O}\left(\left(\|\mathcal{L}_d\|_Lt+\frac{\log(\varepsilon^{-1})}{\log\log(\varepsilon^{-1})}\right)\log(M)\right).
\end{eqnarray}
\end{theorem}

For fast forwarding, we add additional structures to jump operators in Eq. \ref{dlme}. We set $F_i$ to be block-diagonal Pauli operators of the form 
\begin{eqnarray}\label{speciallme}
F_i=\sum_{j=0}^{R-1}|j\rangle\langle j|\otimes P_{j,i}
\end{eqnarray}
with $P_{j,i}$ as $n$-qubit Pauli operators from $\{\pm1,\pm i\}\times\{I,X,Y,Z\}^{\otimes n}$. For these special $F_i$, we introduce another unitary operator $V_F$ that encodes their information more effectively, to take the role of $U_F$,
\begin{eqnarray}\label{vf}
V_{F}:~V_F|i\rangle|0\rangle^{\otimes 4Rn}=|i\rangle|x_{F_i}\rangle,
\end{eqnarray}
where $x_{F_i}$ is the binary representation of $F_i$, constructed in a natural way according to the following encoding: 
\begin{eqnarray}\label{binary}
\renewcommand{\arraystretch}{1.5}
\begin{tabular}{|>{\centering\arraybackslash}p{1cm}
              |>{\centering\arraybackslash}p{1cm}
              |>{\centering\arraybackslash}p{1cm}
              |>{\centering\arraybackslash}p{1cm}
              |>{\centering\arraybackslash}p{1cm}|}
\hline
1 & $i$ & $-1$ & $-i$ \\
\hline
00 & 01 & 10 & 11\\
\hline
\end{tabular}
\quad \quad 
\begin{tabular}{|>{\centering\arraybackslash}p{1cm}
              |>{\centering\arraybackslash}p{1cm}
              |>{\centering\arraybackslash}p{1cm}
              |>{\centering\arraybackslash}p{1cm}
              |>{\centering\arraybackslash}p{1cm}|}
\hline
$I$ & $X$ & $Y$ & $Z$ \\
\hline
00 & 01 & 10 & 11\\
\hline
\end{tabular}
\renewcommand{\arraystretch}{1.0}
\end{eqnarray}
For example, if $F=\ket{0}\bra{0}\otimes (-X) + \ket{1}\bra{1}\otimes (i Z)$, then $x_F = 10010111$. Since now each single qubit Pauli from $\{\pm1,\pm i\}\times\{I,X,Y,Z\}$ can be represented by 4 bits and each $n$-qubit Pauli operator is a tensor product of $n$ such single qubit Pauli operators,  we can use $4n$ bits for the $n$-qubit Pauli representation. For convenience in the algorithm design, we overuse $2n$ bits to encode the global phase, even though only a single phase is associated with each $P_{j,i}$. With this, we overall use $4Rn$ bits.

Using $U_g$ and $V_F$, the complexity of simulating these special Lindblaidans is summarized in Theorem \ref{mainthe2}, whose proof is given in Appendix \ref{sec: appendix B}.

\begin{theorem}\label{mainthe2}
There exists a quantum algorithm that outputs an operator $\mathcal{A}$, for the simulation of the Eq. \ref{dlme} with jump operators satisfying Eq. \ref{speciallme}, such that $\|e^{\mathcal{L}_dt}-\mathcal{A}\|_\diamond\leq \varepsilon$, with
\begin{eqnarray}
\text{queries on $U_g$ and $V_F$:}\quad && \mathcal{O}\left(\|\mathcal{L}_d\|_Lt+\frac{\log(\varepsilon^{-1})}{\log\log(\varepsilon^{-1})}\right).\\
\text{circuit depth:} \quad && \mathcal{O}\left(M\log\left(\|\mathcal{L}_d\|_Lt+\frac{\log(\varepsilon^{-1})}{\log\log(\varepsilon^{-1})}\right)+R\right).\\
\text{ancilla qubits:} \quad&& \mathcal{O}\left(\left(\|\mathcal{L}_d\|_Lt+\frac{\log(\varepsilon^{-1})}{\log\log(\varepsilon^{-1})}\right)\left(\log(M)+Rn\right)\right).\label{aaaddd}
\end{eqnarray}
\end{theorem}


\textbf{Discussion of results: }Unitary channels have been well studied and shown to play important roles in quantum information and quantum computing \cite{gutoski2014process,haah2023query,audenaert2008random,hu2024probabilistic}, with the Pauli channel being a typical example for modeling noise in quantum error correction~\cite{lidar2013quantum}.  Theorem~\ref{mainthe1} shows that, instead of doing general Lindbladian simulations, resulting in multiplicative complexity $\tilde{\mathcal{O}}(t\,\polylog(\varepsilon^{-1}))$ ($\tilde{\mathcal{O}}$ hides $\polylog(t)$ factors) in previous state-of-the-art works, if we only focus on special Lindbladians, i.e. purely dissipative Lindbladians with unitary jump operators (Eq.~\ref{dlme}), we can achieve an improved additive complexity $\mathcal{O}(t+\log(\varepsilon^{-1}))$ with a strictly linear dependence on $t$. Also, while our results in Theorem~\ref{mainthe1} seem similar to the results in QSP-based Hamiltonian simulation \cite{low2017optimal}, however, our method, as will be clearer in the following method overview section, is actually closer to the LCU \cite{childs2012hamiltonian} technique in a modified sense. In Appendix \ref{sec: appendix B}, we also discuss generalizations of Theorem \ref{mainthe1} inspired by \cite{borras2025quantum}. To be more exact, we show that the same complexity can be achieved for a much broader class of Lindbladians as long as 
\begin{eqnarray}\label{general}
\sum_{i=1}^M g_i F_i^\dag F_i\propto I,
\end{eqnarray}
which can be applied to non-unitary jump operators (e.g., in the 1-qubit case: $F_1=|0\rangle\langle 1|  $ and $F_2=|1\rangle\langle 0|$).

We remark that the exponential improvement is in circuit depth, not in query complexity, and that it comes with ``large ancilla usage" and parallel oracle access. However, the fast forwarding for the Lindbladian Eq.~\ref{dlme} with block-diagonal Pauli jump operators Eq.~\ref{speciallme}, as shown in Theorem~\ref{mainthe2}, is non-trivial compared to the Hamiltonian simulation cases. Note that the underlying reason for exponential fast forwarding of Hamiltonians~\cite{atia2017fast,gu2021fast}---such as commuting local Hamiltonians and quadratic fermionic Hamiltonians---is that they are efficiently diagonalizable. In contrast, as evidenced by Theorem~\ref{mainthe3} below, the Lindbladian dynamics considered here cannot be efficiently diagonalized for $R \geq 2$ when viewed in the vectorization picture (see Appendix~\ref{provethe2}), and the Clifford hierarchy increases as $R$ grows~\cite{cui2017diagonal}. As we will show later, the main reason for fast forwarding in this case is the ability to compute the product of $K$ Pauli operators in $\log(K)$ circuit depth. Regarding the physical meaning of block diagonal Pauli noise, when $R=1$, Theorem~\ref{mainthe2} directly applies to the natural Pauli noise. For larger $R$, we can understand them as selected dephasing noise: $F_i$ introduces the dephasing $P_{j,i}$ on the subspace labeled by $|j\rangle$ basis of the first register.

These different origins for exponential fast forwarding lead to different complexity results. In Theorem~\ref{mainthe2}, exponential fast forwarding appears in circuit depth, while the query complexity remains unchanged compared to Theorem~\ref{mainthe1}. As will become clear later, this exponential separation between query complexity and circuit depth arises from the parallel access to $V_F$, meaning increasing the number of ancilla qubits (Eq.~\ref{aaaddd}). We can also understand Theorem~\ref{mainthe2} as a conservation of space-time complexity. Our algorithm suggests that for the special Lindbladians we considered here, one can transfer the time complexity (circuit depth) to the space complexity (increasing ancilla qubits) without compromise, i.e., without increasing the overall space-time (query) complexity. This conservation of space-time complexity is surprising since in most of the situations we met in quantum computing, reducing circuit depth by increasing qubit numbers always implies an increase in the overall space-time complexity. One example is exactly the Hamiltonian simulation. The no parallel fast-forwarding theorem in Ref~\cite{zhang2024parallel} suggests that to reduce the circuit depth for general Hamiltonian simulation, an exponentially larger number of ancilla qubits is needed. Another example is in quantum searching and quantum phase estimation~\cite{kitaev1995quantum}. When we use naive repeated measurements and the Hadamard test~\cite{aharonov1998quantum} (equivalent to increasing space complexity) to replace Grover's algorithm~\cite{grover1996fast} and quantum phase estimation, we will sacrifice the quadratic quantum speedup~\cite{burchard2019lower}, and the Heisenberg limit will go down to the standard quantum limit~\cite{ni2023low}.

\subsection{Estimating Gibbs coherence amplitude}

Given an $n$-qubit Hamiltonian
\begin{eqnarray}\label{mainhami}
H=\sum_{i=1}^M\lambda_i Q_i,
\end{eqnarray}
where $\lambda_i\geq0$ and $Q_i$ are $n$-qubit Pauli operators from $\{\pm1\}\times \{I,X,Y,Z\}^{\otimes n}$, its Gibbs state at inverse temperature $\beta$ is defined as
\begin{eqnarray}
\rho_\beta=\frac{e^{-\beta H}}{\mathcal{Z}(\beta)},
\end{eqnarray}
with $\mathcal{Z}(\beta)=\Tr(e^{-\beta H})$ the partition function. Without loss of generality, we will set $\sum_{i=1}^M\lambda_i=1$.

In this work, we consider the problem of estimating the quantity $\langle \psi_1|e^{-\beta (H+I)}|\psi_2\rangle$ for arbitrary $|\psi_1\rangle$ and $|\psi_2\rangle$, which we call the Gibbs coherence amplitude (GCA). We let $|\psi_1\rangle=U_1|+\rangle^{\otimes n}$ be prepared by a quantum circuit $U_1$ from $|+\rangle^{\otimes n}$ and let $|\psi_2\rangle=U_2|0\rangle^{\otimes n}$ be prepared by a quantum circuit $U_2$ from $|0\rangle^{\otimes n}$. For the convenience of the design of our quantum algorithms, we ask $U_1$ and $U_2$ to be constructed from the universal gate set $\{\mathcal{H}, \mathcal{S}, T, CNOT\}$ with a circuit depth $D_1$ and $D_2$ respectively, and we define $D=D_1+D_2$. We use $n_{h_1}$ and $n_{h_2}$ to denote the number of Hadamard gates in $U_1$ and $U_2$ respectively and define $n_h=n_{h_1}+n_{h_2}$. We assume the Hamiltonian Pauli decomposition in Eq. \ref{mainhami} and the gate decomposition of $U_1$ and $U_2$ are known a priori, whose information will be utilized to construct a quantum algorithm for our purpose with complexity summarized in Theorem \ref{mainthe3}. The detailed analysis can be found in Appendix \ref{sec: app c}.

\begin{theorem}\label{mainthe3}
There exists a quantum algorithm that can estimate $ \langle \psi_1| e^{-\beta(H+I)}|\psi_2\rangle $ up to additive error $\epsilon$ with success probability at least $1-\delta$, and with circuit depth
\begin{eqnarray}
\tilde{\mathcal{O}}\left(2^{-(n-n_h)/2}\epsilon^{-1}\left(M\log(\beta)+D\right)\log\left(\delta^{-1}\right)\right).
\end{eqnarray}

\end{theorem}

%
\textbf{Discussion of results:} We now compare our result in Theorem~\ref{mainthe3} with the one based on QSVT (see Appendix~\ref{qsvt} for analysis). From Table~\ref{tablecom}, we observe that our algorithm achieves exponential improvements in two parameters compared to the QSVT-based approach—namely, a factor of $2^{-n/2}$ versus $1$, and $\log(\beta)$ versus $\sqrt{\beta}$—when estimating GCAs such as $\langle 0|^{\otimes n} e^{-\beta (H+I)} |+\rangle^{\otimes n}$. Note that the lower bound on $\beta$ in terms of query complexity is $\Omega(\sqrt{\beta})$ as shown in Ref~\cite[Theorem 1.5]{montanaro2024quantum}, which comes from the quadratically fast polynomial approximation to monomial functions~\cite[Theorems 4.1 and 5.3]{sachdeva2014faster}. Our result actually obeys this lower bound since our fast-forwarding Lindbladian simulation algorithm, Theorem~\ref{mainthe2}, is only for circuit depth. We remark that if the spectral norm $\|H\|$ is known, the QSVT approach allows for spectral amplification~\cite{low2017hamiltonian}, which can mitigate the exponentially decaying factor $e^{-\beta}$, yielding exponentially better scaling in $\beta$ compared to our approach (Appendix~\ref{qsvt}), and our approach will then only have the $2^{-n/2}$ advantage. However, in general, it is infeasible to determine the spectral norm of $H$ in advance, as estimating the ground energy of quantum Hamiltonians is known to be QMA-hard~\cite{kempe2004complexity}.
\begin{table}[htbp]
\renewcommand{\arraystretch}{2}
\centering
\begin{tabular}{|c|c|}\hline
\textbf{Algorithms} & \textbf{Complexity of estimating GCA} \\\hline 
QSVT (Ref. \cite{gilyen2019quantum}) &  $\tilde{\mathcal{O}}\left(\epsilon^{-1}\left(M\sqrt{\beta}+D\right)\log\left(\delta^{-1}\right)\right)$ \\\hline
~~~~This work (Theorem \ref{mainthe3})~~~~ &~~~~  $\mathcal{O}\left(2^{-(n-n_h)/2}\epsilon^{-1}\left(M\log(\beta)+D\right)\log\left(\delta^{-1}\right)\right)$~~~~ \\\hline
\end{tabular}
\caption{Comparison of our algorithm to the QSVT-based algorithm, with complexity measured by circuit depth.\label{tablecom}}
\end{table}

We can also compare our approach of using Lindbladians for Gibbs-related problems with recently proposed Lindbladian-based Gibbs state preparation algorithms~\cite{chen2023quantum, chen2023efficient, ding2025efficient, rouze2024optimal}. Unlike those approaches, which encode Gibbs states as the steady states of Lindbladians, our method encodes the Gibbs operator $e^{-\beta H}$ into the Lindbladian dynamics. 
The main limitation of our approach is that it does not allow full preparation of the Gibbs state or support sampling tasks~\cite{rajakumar2024gibbs, bergamaschi2024quantum}. However, for the task of GCA estimation, our method appears to offer significant advantages, which can be summarized in three aspects.
First, estimating GCA from the Gibbs state $\rho_\beta$ requires knowledge of the partition function $\Tr(e^{-\beta(H+I)})$, since
$\langle \psi_1 | \rho_\beta | \psi_2 \rangle = \frac{\langle \psi_1 | e^{-\beta(H+I)} | \psi_2 \rangle}{\Tr(e^{-\beta(H+I)})},$
as discussed in~\cite{bravyi2021complexity}.
Second, the complexity of these Gibbs state preparation algorithms depends on the mixing time of the Lindbladians, which is generally expected to scale exponentially with $\beta$ due to the inherent hardness of ground state preparation.
Finally, the complexity reduction of $\mathcal{O}(2^{-(n - n_h)/2})$ is unlikely to be seen their approach.

We want to emphasize that, here, while we use additive error $\epsilon$ in our complexity presentation, we can always convert it into relative error esitimation via the relation that $\epsilon$-relative error in $\langle \psi_1|e^{-\beta (H+I)}|\psi_2\rangle$ corresponds to $|\langle \psi_1|e^{-\beta (H+I)}|\psi_2\rangle|\epsilon$ additive error.

\textbf{Application: }While GCA is an artificial quantity to fit our algorithm best, it is actually related to the physically-interested property called the complex-time Loschmidt echo (or the boundary partition function): $ \langle \psi_1| e^{-(\beta+it)(H)}|\psi_2\rangle $. (Note that $e^{-iHt}$ can be absorbed into $U_1$ or $U_2$.) The normalized logarithm of the Loschmidt echo: $-2^{-n}\log(\langle \psi_1| e^{-(\beta+it)(H)}|\psi_2\rangle )$ is known as the free energy density. The non-analytic behavior of the free energy density can diagnose the quantum phase transitions in the complex-time domain (real part for thermal transition and imaginary part for dynamical transition). We recommend Ref.~\cite{heyl2013dynamical,heyl2018dynamical,wei2025riemann,leclair1995boundary,santilli2020phase} for further explanations. Below, we focus especially on two regimes, $\beta=0$ for amplitude estimation and large $\beta$ for ground state overlap testing.

At infinite temperature (i.e., $\beta = 0$), estimating GCA becomes the standard amplitude estimation problem (i.e., estimating $\langle \psi_1 | \psi_2 \rangle$), which has a known lower bound of $\Omega(\epsilon^{-1} D \log(\delta^{-1}))$ on the total cost~\cite{zwierz2010general, bennett1997strengths}. This is also the lower bound on circuit depth when considering the Heisenberg limit~\cite{burchard2019lower}.
In comparison, our algorithm in Theorem \ref{mainthe3} achieves a complexity of 
$\mathcal{O}(2^{-(n-n_h)/2}\epsilon^{-1}D\log\left(\delta^{-1}\right))$
which can exponentially improve the complexity in the regime where $n - n_h = \Theta(n)$---for instance, when $n_h = \sqrt{n}$ or $n_h = 3n/4$. Since a linear number of Hadamard gates is allowed, with no restrictions on magic states~\cite{gottesman1998heisenberg} or entanglement~\cite{horodecki2009quantum}, to our knowledge, no existing classical algorithms can efficiently simulate quantum systems in this regime \cite{aaronson2004improved, orus2014practical}. We want to mention that the result in Theorem~\ref{mainthe3} also indicates that the absolute value of the estimated GCA is upper bounded by $2^{-(n-n_h)/2}$, which will be clear in Appendix \ref{sec: app c}. 

The maximal improvement of our algorithm is $2^{-n/2}$, which corresponds to the case where $|\psi_1\rangle$ is a $+/-$ basis and $|\psi_2\rangle$ is a $0/1$ basis. This result highlights the importance of optimizing quantum algorithms for special cases, rather than focusing on worst-case complexity lower bounds.
Interestingly, the dependence on the number of Hadamard gates is particularly noteworthy. Hadamard gates can generate quantum coherence~\cite{streltsov2017colloquium} and, when combined with Toffoli gates, can encode the GapP function—which is related to the number of solutions to polynomial equations over $\mathbb{Z}_2$—into quantum amplitudes~\cite{dawson2004quantum}. Therefore, at infinite temperature, the improvement shown here may indicate a super-quadratic quantum speedup for special GapP instances, although classical algorithms for these cases still require careful analysis. We summarize the complexity result of Theorem \ref{mainthe3} at $\beta=0$ in the following corollary.
\begin{corollary}\label{corogcq}
Given a unitary circuit $U$ composed of $\{\mathcal{H}, \mathcal{S}, T, CNOT\}$ with a circuit depth $D$, there exists a quantum algorithm that can estimate $ \langle +|^{\otimes n} U|0\rangle^{\otimes n} $ up to additive error $\epsilon$ with success probability at least $1-\delta$, and whose circuit depth is
\begin{eqnarray}
\mathcal{O}\left(2^{-(n-n_h)/2}\epsilon^{-1}\log\left(\delta^{-1}\right)D\right).
\end{eqnarray}
\end{corollary}

Besides treating $e^{-\beta H}$ as a Gibbs operator, we can also view it as imaginary-time evolution.  Therefore, as long as there is a non-zero overlap between the ground state $|G\rangle$ of $H$ and $|\psi_2\rangle$, the state $e^{-\beta (H+I)}|\psi_2\rangle$ is close to $|G\rangle$ for large enough $\beta$. Thus, our GCA estimation algorithm can be used to probe ground-state properties: how close the ground state $|G\rangle$ is to a given reference state $|\psi_1\rangle$. We consider a practical example of estimating $\langle 0|^{\otimes n}|G\rangle$ with the following corollary
\begin{corollary}
Starting from $|\psi_2\rangle=|+\rangle^{\otimes n}$ with the overlap with $|G\rangle$ is $c_0$, then, to estimate $\langle 0|^{\otimes n}|G\rangle$ within an additive error $\epsilon$ with success probability at least $1-\delta$, our algorithm achieves a circuit depth
\begin{eqnarray}
\mathcal{O}\left(|c_0|^{-1}\epsilon^{-1}\log(\delta^{-1})M2^{-n/2}(\epsilon^{-1}|c_0|^{-1})^{\frac{1+E_0}{\Delta}}\right),
\end{eqnarray}
with $\Delta$ the gap of $H$ and $E_0$ the ground energy ($<0$).
\end{corollary}
The proof can be found in Appendix \ref{gsot}. This is actually a common and meaningful setting: the equally superposed state $|\psi_2\rangle$ is used as the initial state for ground-state preparation, and the ground-state overlap is tested on the computational basis state $|\psi_1\rangle$. Note that the distribution of a state over the computational basis affects the complexity of the sampling problem \cite{hangleiter2023computational}. In comparison, using the state-of-the-art ground state preparation algorithm \cite{lin2020near}, the complexity is $\mathcal{O}\left(|c_0|^{-1}\epsilon^{-1}\log(\delta^{-1})M\Delta^{-1}\right)$. In Appendix \ref{gsot}, we show that our complexity can have exponential advantages over the state-of-the-art one when the Hamiltonian frustration $1+E_0$ is smaller than the Hamiltonian gap $\Delta$ for guided initial state ($|c_0|^{-1}=\mathcal{O}(\text{poly} (n))$), and can have exponential advantages over the state-of-the-art one when the Hamiltonian frustration $1+E_0$ is smaller than the half of the Hamiltonian gap $\Delta/2$ for random initial state the makes $|c_0|$ around $2^{-n/2}$.

\section{Methods overview}

In this section, we briefly give an overview of the main ideas behind our algorithms, with more details deferred to the appendices.
Before the overview, we introduce the notion of vectorization. The mapping from an $n$-qubit operator $O = \sum_{ij} o_{ij} |i\rangle\langle j|$ to a $2n$-qubit vector $\ket{O} = \sum_{ij} o_{ij} |i\rangle |j\rangle$ is called \emph{vectorization}. One can also interpret $\ket{O}$ as the Choi vector of $O$, e.g., see~\cite{choi1975completely}. Any general linear transformation $O \rightarrow \sum_i A_i O B_i^\dag$ can be re-expressed in the vectorized picture as $\ket{O} \rightarrow \sum_i A_i \otimes B_i^* \ket{O}$. Here, $B_i^\dag$ and $B_i^*$ denote the Hermitian conjugate and complex conjugate of $B_i$, respectively.

\subsection{Lindbladian fast forwarding}

When Lindbladians are purely dissipative and the jump operators are unitaries (Eq.~\ref{dlme}), the solution at time $t$ in the vectorization picture has the form
\begin{eqnarray}
\ket{\rho(t)} = e^{-T} \exp\left(T \sum_{i=1}^M p_i F_i \otimes F_i^* \right) \ket{\rho(0)},
\end{eqnarray}
where $p_i = g_i / \|\mathcal{L}_d\|_L$ and $T = \|\mathcal{L}_d\|_L t$. The exact dynamics $\exp(T \sum_{i=1}^M p_i F_i \otimes F_i^* )$ can be well approximated, up to error $\varepsilon$, by a Taylor truncation of degree \cite{gilyen2019quantum} 
\begin{equation}
K = \mathcal{O} \left( T + \frac{\log(\varepsilon^{-1})}{\log\log(\varepsilon^{-1})} \right).
\end{equation}
The Taylor expansion takes the form
\begin{eqnarray}\label{eqn:Taylor}
\ket{\rho(t)} 
&\approx& e^{-T} \sum_{k=0}^K \frac{T^k}{k!} \left( \sum_{i=1}^M p_i F_i \otimes F_i^* \right)^k \ket{\rho(0)} \nonumber \\
&=& e^{-T} \sum_{k=0}^K \frac{T^k}{k!} \sum_{\vec{i}_k} p_{\vec{i}_k} F_{\vec{i}_k} \otimes F_{\vec{i}_k}^* \ket{\rho(0)},
\end{eqnarray}
where $\sum_{\vec{i}_k} p_{\vec{i}_k} F_{\vec{i}_k} \otimes F_{\vec{i}_k}^*$ denotes the expansion of $( \sum_{i=1}^M p_i F_i \otimes F_i^*)^k$ (see the formal definition in Appendix~\ref{proofthe1}). Since $F_i$ are unitaries, the Taylor truncation in the vectorization picture takes the form of an LCU, and corresponds to a quantum channel with unitary Kraus operators in the original (operator) picture.

A purification of this quantum channel can be implemented by a quantum circuit composed of $U_g$ (Eq.~\ref{mainug}) and $U_F$ (Eq.~\ref{uf}). Both the query complexities for $U_g$ and $U_F$ are $\mathcal{O}(K)$, and the queries to $U_F$ are made adaptively. Since $U_F$ is a control operator including $M+1$ terms, given access to all $F_i$ operators, the circuit depth to implement $U_F$ is $\mathcal{O}(M)$~\cite{nielsen2010quantum} (which is also true for $U_g$). As a result, the overall circuit depth becomes
\begin{equation}\label{eqn:complexity}
\mathcal{O} \left( M \left( T + \frac{\log(\varepsilon^{-1})}{\log\log(\varepsilon^{-1})} \right) \right).
\end{equation}
This establishes the main idea of proving Theorem \ref{mainthe1}. Further details, including the number of ancillas are provided in Appendix~\ref{proofthe1}.

As shown in Eq.~\ref{eqn:complexity}, our algorithm has an additive scaling between $T$ and $\varepsilon$. 
This is a remarkable feature and, to the extent of our knowledge, is the first time for an LCU-like simulation algorithm to achieve additive scaling. 
Specifically, existing LCU-based algorithms for simulating Hamiltonians~\cite{berry2015simulating} or Lindbladians~\cite{li2022simulating} also first approximates the evolution operator by a truncated Taylor or Dyson series, and then attempts to implement the truncated series by LCU. 
However, if the expansion is performed over a long time interval, then the $1$-norm of the coefficients could be exponential in $T$, leading to an inefficient LCU implementation. 
To resolve this issue, existing algorithms breaks the long time interval into $T$ many short ones, implementing the truncated series by LCU on each short interval, and applies an oblivious amplitude amplification \cite{berry2015simulating} technique in each step to ensure that the success probability does not decay exponentially. 
Although this is a time-efficient approach, the overall complexity becomes multiplicative, as the overall cost would be the cost per step (which is $\mathcal{O}(\log(T/\varepsilon))$) times the overall number of steps (which is $\mathcal{O}(T)$). 
As a comparison, our algorithm does not need to perform such time discretization and can directly implement the long-time Taylor series. 
This is because the Lindcladians we consider are purely dissipative, leading to an extra rescaling factor $e^{-T}$ in the expansion in Eq.~\ref{eqn:Taylor} and a constant $1$-norm of the coefficients. 
Therefore our algorithm avoids the time discretization and oblivious amplitude amplification steps, is a global algorithm and thus has an additive scaling.

For the fast-forwarding (Theorem~\ref{mainthe2}), note that each unitary in the LCU form of the Taylor truncation in the vectorization picture is a product of up to $K$ operators of the form $F_i \otimes F_i^*$. When we restrict each $F_i$ to be of the form given in Eq.~\ref{speciallme}, the result of the product $F_i \otimes F_i^*$ can be efficiently inferred using the multiplication rules of Pauli matrices. To realize the corresponding (purification of the) quantum channel on a quantum computer, assuming Eq.~\ref{speciallme}, each Kraus operator can be constructed with up to $\log(K)$ layers of $V_F$ (Eq.~\ref{vf}) through the binary encoding from Eq.~\ref{binary}. Note that, given a classical tableau $i\rightarrow x_{F_i}$, the circuit depth of constructing $V_F$ is also $\mathcal{O}(M)$ by converting the reversible classical circuit into the quantum circuit~\cite{nielsen2010quantum}. The reduction in depth from $K$ to $\log(K)$ arises from the observation that the product of $K$ Pauli operators can be computed in logarithmic depth: first group the $K$ operators into $K/2$ pairs and multiply them in parallel, then group the resulting $K/2$ operators into $K/4$ pairs, and so on.
As a result, while the number of queries to $V_F$ remains $K$, the circuit depth can be reduced to $\mathcal{O}(\log K)$. Note that each $F_i$ is not directly a single Pauli operator but consists of $R$ Pauli blocks; hence, the parameter $R$ also contributes to the overall simulation complexity.

Combining these components, the total circuit depth of the fast-forwarding algorithm is
\begin{equation}
\mathcal{O} \left( M \log \left( T + \frac{\log(\varepsilon^{-1})}{\log\log(\varepsilon^{-1})} \right) + R \right).
\end{equation}
See Appendix~\ref{provethe2} for further details about the proof of Theorem \ref{mainthe2}, including the number of ancillas.

\subsection{Estimating Gibbs coherence amplitude}

Consider the following $(n+1)$-qubit quantum channel:
\begin{equation}
\mathcal{C}[\rho] = \sum_i 
\begin{pmatrix}
K_i & 0 \\
0 & L_i
\end{pmatrix}
\rho 
\begin{pmatrix}
K_i^\dag & 0 \\
0 & L_i^\dag
\end{pmatrix},
\end{equation}
where $\left\{
\begin{pmatrix}
K_i & 0 \\
0 & L_i
\end{pmatrix}
\right\}$ is a set of Kraus operators. When $\rho_S := \begin{pmatrix}
\cdot & S \\
S^\dag & \cdot
\end{pmatrix}$, we have 
\begin{equation}
\rho_{\mathrm{out}} := \mathcal{C}[\rho_S] = 
\begin{pmatrix}
\cdot & \sum_i K_i S L_i^\dag \\
\sum_i L_i S^\dag K_i^\dag & \cdot
\end{pmatrix},
\end{equation}
from which, we can find that if we focus on the upper-right block, we can get rid of the CPTP restrictions. The expectation value of the operator $X \otimes I_n$ with respect to $\rho_{\mathrm{out}}$ is
\begin{equation}\label{eamp}
\Tr\big((X \otimes I_n)\rho_{\mathrm{out}}\big) = 2 \, \mathrm{Re} \left[ \Tr\left( \sum_i K_i S L_i^\dag \right) \right] 
= 2^{n/2 + 1} \, \mathrm{Re} \left[ \langle \Omega | \sum_i K_i \otimes L_i^* \ket{S\rangle} \right],
\end{equation}
where $\ket{\Omega} = 2^{-n/2} \sum_i \ket{ii}$ is the $2n$-qubit Bell state. To access the imaginary part, we can consider the observable $Y \otimes I_n$.

Our algorithm for estimating $\langle \psi_1 | e^{-\beta (H + I_n)} | \psi_2 \rangle$ is based on constructing a quantum channel such that 
\begin{equation}
\langle \Omega | \sum_i K_i \otimes L_i^* \ket{S\rangle} \propto \langle \psi_1 | e^{-\beta (H + I_n)} | \psi_2 \rangle.
\end{equation}
Due to vectorization, the identity operator can be interpreted as the Bell state with exponentially large norm $2^{n/2}$. Therefore, if we estimate the left-hand side of Eq.~\ref{eamp}, the result will be exponentially amplified, reducing the estimation complexity. This is the origin of the $2^{-n/2}$ factor in Theorem~\ref{mainthe3}. However, the norm of $\sum_i K_i \otimes L_i^* \ket{S\rangle}$ can be small due to the dissipative nature of quantum channels.

We refer to $\rho_S$ as the {non-diagonal density matrix encoding} (NDME) of $S$, and the way of encoding operations into $\sum_i K_i \otimes L_i^*$ as the channel block encoding (CBE) \cite{shang2024estimating,shang2024design} (see formal definitions in Appendix~\ref{esgca}). In our setting, we adopt the Pauli encoding where $I$ is mapped to $\ket{0}$ and $X$ to $\ket{1}$. For example, when $S \propto I + X$, the vectorized form $\ket{S\rangle}$ corresponds to $\ket{+}$. This encoding choice significantly simplifies our analysis and, more importantly, allows us to relate the norm $\left\| \sum_i K_i \otimes L_i^* \ket{S\rangle} \right\|$ to the quantum coherence resource \cite{streltsov2017colloquium}, giving rise to the $2^{n_h/2}$ factor in Theorem~\ref{mainthe3}.

More concretely, since GCA takes the form $\langle +|^{\otimes n} U_1^\dag e^{-\beta (H + I_n)} U_2 \ket{0}^{\otimes n}$, we construct the CBE of $U_1^\dag e^{-\beta (H + I_n)} U_2$ using three successive channels $\mathcal{C}_{U_1^\dag} \circ e^{\mathcal{L}_H \beta} \circ \mathcal{C}_{U_2}.$
For $\mathcal{C}_{U_1^\dag}$ and $\mathcal{C}_{U_2}$, since $U_1$ and $U_2$ are composed from basic gates in the universal set $\{\mathcal{H}, \mathcal{S}, T, \mathrm{CNOT}\}$, we first build basic channels corresponding to each gate and then compose them to form the desired channel. As a result, the circuit depths of $\mathcal{C}_{U_1^\dag}$ and $\mathcal{C}_{U_2}$ are $D_1$ and $D_2$, respectively. Under the Pauli encoding, the Hadamard gate becomes special since we prove that the optimal CBE construction can only encode $\mathcal{H}/\sqrt{2}$, while for other gates, 1 to 1 encoding is possible. Therefore, $n_h$ Hadamard gates will exponentially ($2^{-n_h/2}$) dissipate the norm of $\sum_i K_i \otimes L_i^* \ket{S\rangle}$.

For $e^{\mathcal{L}_H \beta}$, which provides a CBE of $e^{-\beta (H + I_n)}$, we directly utilize fast-forwardable Lindbladians with $R = 2$. For any Hamiltonian $H$, we show that there exists such a Lindbladian with carefully chosen jump operators, such that its evolution for time $\beta$ simulates a CBE of $e^{-\beta (H + I_n)}$. As demonstrated in Theorem~\ref{mainthe2}, this simulation has circuit depth $\mathcal{O}(M \log \beta)$.

After preparing (a purification of) $\rho_{\mathrm{out}}$, we apply the amplitude estimation algorithm \cite{brassard2000quantum,aaronson2020quantum,grinko2021iterative,rall2023amplitude} to estimate either $\Tr((X \otimes I_n)\rho_{\mathrm{out}})$ or $\Tr((Y \otimes I_n)\rho_{\mathrm{out}})$ with Heisenberg-limit scaling. Putting all components together yields the total circuit depth stated in Theorem~\ref{mainthe3}. See Appendix~\ref{proofthe3} for full details.

\section{Summary and outlook}

In summary, we found that for purely dissipative Lindbladians with unitary jump operators, their simulation can be achieved with additive complexity, as stated in Theorem~\ref{mainthe1}, improving upon the multiplicative complexity of existing Lindbladian simulation algorithms. Furthermore, when the jump operators are restricted to block-diagonal Pauli operators, we presented a quantum simulation algorithm that enables exponential fast-forwarding in circuit depth while preserving the query complexity (Theorem~\ref{mainthe2}). To the best of our knowledge, this is the first result on Lindbladian fast-forwarding.

An application of these fast-forwarding Lindbladians is the estimation of Gibbs state properties, such as the Gibbs Coherence Amplifier (GCA). By leveraging fast-forwarded Lindbladian simulations, we encode the GCA into the non-diagonal blocks of density matrices using several new techniques, including Non-Diagonal Matrix Encoding (NDME), Coherence-Based Encoding (CBE), and Pauli encoding. Under certain conditions related to quantum coherence, the encoded GCA can be exponentially amplified. As a result, compared with the advanced QSVT-based approaches for estimating GCA, our algorithm can achieve exponential advantages in two parameters (Theorem~\ref{mainthe3}).

There are three directions that merit further investigation. 
\begin{itemize}
    \item First, it is important to identify more classes of Lindbladians that support fast-forwarding, to explore the general conditions enabling Lindbladian fast-forwarding, and to understand how these differ from their Hamiltonian counterparts. 
    \item Second, it is also important to uncover the physical significance and potential applications of GCA, particularly in regimes where our algorithm offers exponential advantages. Furthermore, by replacing the Pauli encoding with alternative encoding schemes, we can also explore what other types of quantum resources—beyond coherence—may be relevant to the complexity of estimating GCA. 
    \item Third, it is interesting to explore whether the new techniques introduced in this work, including both the fast-forwarding and vectorization techniques (e.g., NDME, CBE), can inspire the discovery of new quantum algorithms and novel applications.
\end{itemize}

\begin{acknowledgments}
The authors would like to thank Joao Doriguello for his helpful comments on an early version of this paper. ZS would also like to thank Daniel Stilck Franca for valuable discussions, and would like to thank Naixu Guo, Tongyang Li, Patrick Rebentrost, and Qi Zhao for inspiring discussions on related topics. DA acknowledges funding from Innovation Program for Quantum Science and Technology via Project 2024ZD0301900, and the support by The Fundamental Research Funds for the Central Universities, Peking University. ZS acknowledges support from HK institute of Quantum Science and Technology. CS is supported by the National Key Research Project of China under Grant No. 2023YFA1009403. 
\end{acknowledgments}

\textbf{Note added:} Shortly after the initial posting of this manuscript on arXiv, we became aware of independent work by Borras et al. \cite{borras2025quantum}, which also establishes additive complexity for the Lindbladians considered in Theorem \ref{mainthe1} using a distinct algorithmic approach. We note that Ref. \cite{borras2025quantum} focuses primarily on additive complexity simulation and introduces the Hamiltonian part, while our work focus additionally on Lindbladian fast-forwarding and applications to GCA.

Ref~\cite{cleve1998quantum,zhang2022quantum,grover2002creating,preskill1998lecture,michalakis2013stability,cade2022improved,mele2024introduction} are cited in Appendix.

\bibliography{ref}

\clearpage

\begin{appendix}
\onecolumngrid
\renewcommand{\addcontentsline}{\oldacl}
\renewcommand{\tocname}{Appendix Contents}
\tableofcontents

\section{Preliminaries}

\subsection{Amplitude estimation}

Amplitude estimation \cite{brassard2000quantum,aaronson2020quantum,grinko2021iterative,rall2023amplitude} is an algorithm that can be used to estimate absolute values of quantum amplitudes with the Heisenberg limit. 

\begin{lemma}[Amplitude estimation, absolute value]\label{ae}
Given two quantum states $|S_1\rangle=U_1|0^n\rangle$ and $|S_2\rangle=U_2|0^n\rangle$, the amplitude estimation algorithm can give an $\epsilon$-additive error estimation of $|\langle S_1|S_2\rangle|$ with success probability at least $1-\delta$ by querying 
\begin{eqnarray}
U_G=(I-2U_1|0\rangle\langle 0|U_1^\dag)(I-2U_2|0\rangle\langle 0|U_2^\dag),
\end{eqnarray}
for $\mathcal{O}(\epsilon^{-1}\log(\delta^{-1}))$ times.
\end{lemma}

Note that Lemma \ref{ae} can only estimate absolute values $|\langle S_1|S_2\rangle|$. We now show how to use amplitude estimation to estimate $\re(\langle S_1|S_2\rangle)$ and $\im(\langle S_1|S_2\rangle)$. The idea is to use the Hadamard test \cite{cleve1998quantum} to embed $\re(\langle S_1|S_2\rangle)$ and $\im(\langle S_1|S_2\rangle)$ into the absolute values of other amplitudes. We state this in the following lemma formally.

\begin{lemma}[Amplitude estimation, entire estimation]\label{oe}
Given two quantum states $|S_1\rangle=U_1|0^n\rangle$ and $|S_2\rangle=U_2|0^n\rangle$,  there exist quantum algorithms that can give an $\epsilon$-additive error estimation of $\re(\langle S_1|S_2\rangle)$ and $\im(\langle S_1|S_2\rangle)$, with success probability at least $1-\delta$ by querying $U_1$ and $U_2$ (including their reverses) for $\mathcal{O}(\epsilon^{-1}\log(\delta^{-1}))$ times.
\end{lemma}

\begin{proof}
Let $U=U_1^\dag U_2$, we do the following Hadamard test with one ancilla qubit
\begin{eqnarray}
\frac{1}{\sqrt{2}}\left(|0\rangle+|1\rangle\right)\otimes|0^n\rangle
&\xrightarrow[]{\text{control-}U}&\frac{1}{\sqrt{2}}\left(|0\rangle\otimes|0^n\rangle+|1\rangle\otimes U|0^n\rangle\right) \nonumber\\
&\xrightarrow[]{\mathcal{H} \text{ on ancilla}}&|\psi\rangle=\frac{1}{2}\left(|0\rangle\otimes\left(I+U\right)|0^n\rangle+|1\rangle\otimes \left(I-U\right)|0^n\rangle\right).\nonumber
\end{eqnarray}
It is easy to see that
$$\langle\psi|Z\otimes I |\psi\rangle=\frac{1}{2}\langle 0^n|U^\dag+U|0^n\rangle=\re(\langle S_1|S_2\rangle).$$ Now we can add another ancilla qubit and consider the following unitary that encodes $Z\otimes I+I\otimes I$
$$U_Z=\begin{pmatrix}
\frac{1}{2} (Z\otimes I+I\otimes I)	&	\cdot	\\
\cdot	&	\cdot	\\
\end{pmatrix} ,$$
through which we have the relation
$$|\langle 0|\langle\psi| U_Z|0\rangle|\psi\rangle|=\frac{1+\re(\langle S_1|S_2\rangle)}{2} ,$$
where $\re(\langle S_1|S_2\rangle)$ is encoded into the absolute value of the amplitude $|\langle 0|\langle\psi| U_Z|0\rangle|\psi\rangle|$. Therefore, we can use Lemma \ref{ae} for the amplitude estimate by constructing the operator
\begin{eqnarray*}
(I-2|0\rangle|\psi\rangle\langle0|\langle\psi|)(I-2U_Z|0\rangle|\psi\rangle\langle0|\langle\psi|U_Z^\dag),
\end{eqnarray*}
which contains a single query on $U_1$, $U_1^\dag$, $U_2$, $U_2^\dag$ respectively. Combined with the complexity in Lemma \ref{ae}, the claim is proved. For the imaginary part, the proof is similar; one only needs to add a $\mathcal{S}^\dag=\text{diag}(0,-i)$ gate before $\text{control-}U$ in the Hadamard test circuit.
\end{proof}

An important thing needs to be emphasized is that the calls to $U_1$ and $U_2$ in both Lemma \ref{ae} and Lemma \ref{oe} are done adaptively \cite{chia2023impossibility}, so the circuit depth is also equal to $\mathcal{O}(\epsilon^{-1}\log(\delta^{-1}))$.

\subsection{Vectorization}

\begin{definition}[Vectorization and Matrixization]
The mapping from an $n$-qubit matrix $O=\sum_{ij}o_{ij}|i\rangle\langle j|$ to a $2n$-qubit vector $\ket{O\rangle}=\sum_{ij}o_{ij}|i\rangle|j\rangle$ is called vectorization, and is denoted by $\mathcal{V}[O]=\ket{O\rangle}$. The reverse mapping from $\ket{O\rangle}$ to $O$ is called matrixization, and is denoted by $\mathcal{M}[\ket{O\rangle}]=O$.
\end{definition}

One can interpret $\ket{O\rangle}$ as a Choi state \cite{choi1975completely} of $O$. In this work, we use the notation $\ket{\cdot\rangle}$ for unnormalized vectors.
Any linear transformation of the form $O\rightarrow\sum_i A_i O B_i^\dag$ can be re-expressed in the vectorization picture:
\begin{eqnarray}
\mathcal{V}\left[\sum_i A_i O B_i^\dag\right]=\sum_i A_i\otimes B_i^* \ket{O\rangle}.
\end{eqnarray}
The Lindbladian in Eq. \ref{mainlme} can be re-expressed in the vectorization picture as follows
\begin{equation}
\quad\frac{d\ket{\rho(t)\rangle}}{dt}=L\ket{\rho(t)\rangle},
\end{equation}
where $\ket{\rho(t)\rangle}=\sum_{ij}\rho_{ij}(t)|i\rangle|j\rangle$ is the vector representation of the density matrix $\rho(t)$, and $L$ is the Liouvillian generator for the Lindbladian semi-group with matrix form
\begin{eqnarray}
L=-i(H_s\otimes I-I\otimes H_s^T)+\sum_{i=1}^M g_i \left(F_i\otimes F_i^*-\frac{1}{2}F_i^\dag F_i\otimes I-I\otimes\frac{1}{2}F_i^TF_i^*\right).
\end{eqnarray}

\section{Lindbladian simulation} \label{sec: appendix B}

\subsection{Proof of Theorem \ref{mainthe1}\label{proofthe1}}
We consider the simulation of the Lindbladian 
\begin{equation}\label{slme}
\frac{d\rho(t)}{dt}=\mathcal{L}_d[\rho(t)]=\sum_{i=1}^M g_i\left(F_i\rho(t) F_i^\dag-\rho(t)\right),
\end{equation}
where $F_i$ are unitary operators. In the vectorization picture, the Lindbladian has the form
\begin{equation}\label{vslme}
\frac{d\ket{\rho(t)\rangle}}{dt}=L_d\ket{\rho(t)\rangle}=\sum_{i=1}^M g_i\left(F_i\otimes F_i^*-I\right)\ket{\rho(t)\rangle}.
\end{equation}
The solution is
\begin{eqnarray}\label{sol}
\ket{\rho(t)\rangle}&&=\exp(L_dt)\ket{\rho(0)\rangle}=\exp\left(t\sum_{i=1}^M g_i F_i\otimes F_i^* -\sum_{i=1}^M g_it I\right)\ket{\rho(0)\rangle}\nonumber\\&&
=\exp\left(\|\mathcal{L}_d\|_Lt\sum_{i=1}^M p_i F_i\otimes F_i^* -\|\mathcal{L}_d\|_Lt I\right)\ket{\rho(0)\rangle}\nonumber\\
&&=e^{-T}\exp\left(T\sum_{i=1}^M p_i F_i\otimes F_i^* \right)\ket{\rho(0)\rangle},
\end{eqnarray}
where $p_i=g_i/\|\mathcal{L}_d\|_L$ and $T=\|\mathcal{L}_d\|_Lt$.

We can approximate Eq. \ref{sol} via the Taylor truncation
\begin{eqnarray}\label{expan}
\ket{\rho(t)\rangle}&&\approx e^{-T} \sum_{k=0}^K \frac{T^k}{k!} \left(\sum_{i=1}^M p_i F_i\otimes F_i^*\right)^k\ket{\rho(0)\rangle}\nonumber\\&&=e^{-T} \sum_{k=0}^K\frac{T^k}{k!}\sum_{\vec{i_k}}p_{\vec{i_k}} F_{\vec{i_k}}\otimes F_{\vec{i_k}}^* \ket{\rho(0)\rangle},
\end{eqnarray}
where $\sum_{\vec{i_k}}p_{\vec{i_k}} F_{\vec{i_k}}\otimes F_{\vec{i_k}}^*$ is the expansion of $\left( \sum_{i=1}^M p_i F_i\otimes F_i^*\right)^k$ and $K=\mathcal{O}(T + \frac{\log(\varepsilon^{-1})}{\log\log (\varepsilon^{-1})})$. 
For Eq. \ref{expan}, we have two observations. First, the 1-norm is bounded by 1 because
\begin{eqnarray}
e^{-T} \sum_{k=0}^K\frac{T^k}{k!}\sum_{\vec{i_k}}p_{\vec{i_k}} =e^{-T} \sum_{k=0}^K\frac{T^k}{k!}\left( \sum_{i=1}^M p_i \right)^k=e^{-T} \sum_{k=0}^K\frac{T^k}{k!} < e^{-T} e^{T}=1. 
\end{eqnarray}
Second, each term $F_{\vec{i_k}}\otimes F_{\vec{i_k}}^*$ in the vectorization picture corresponds to a unitary transformation in the original picture
\begin{eqnarray}
F_{\vec{i_k}}\rho(0) F_{\vec{i_k}}^\dag.
\end{eqnarray}
Therefore, the Lindbladian dynamics can be directly and approximately interpreted as a quantum channel in the Kraus-sum representation.

Based on Eq.~\ref{expan}, we now present a quantum algorithm for Lindbladian simulation. We consider a $(K + K\log(M+1) + n)$-qubit circuit, where:
\begin{itemize}
    \item The first $K$ qubits represent the first environment, labeled as $a$;
    \item The next $K\log(M+1)$ qubits represent the second environment, which is further decomposed into $K$ groups, each with $\log(M+1)$ qubits, labeled as $b_1, b_2, \dots, b_K$;
    \item The remaining $n$ qubits represent the system governed by the Lindbladian dynamics.
\end{itemize}
The implementation of the Lindbladian simulation requires three different unitary blocks. Without loss of generality, and for simplicity, we assume that the initial state is a pure state: $\rho(0) = |\psi_0\rangle\langle\psi_0|$ for some state $\ket{\psi_0}$.

We start from the initial state $|0\rangle_a |0\rangle_{b_1} \cdots |0\rangle_{b_K} |\psi_0\rangle$, and first apply a unitary circuit $U_T$ on the first environment such that
\begin{equation} \label{state:U_T0}
U_T |0\rangle_a = C \sum_{k=0}^K \sqrt{e^{-T} \frac{T^k}{k!}} \, |\underbrace{1\cdots1}_k \underbrace{0\cdots0}_{K-k}\rangle_a,
\end{equation}
where the normalization constant $C$ is given by
\begin{equation}
C = \sqrt{ \frac{1}{\sum_{k=0}^K e^{-T} \frac{T^k}{k!}} }.
\end{equation}
Implementing $U_T$ by a unitary circuit requires efficient quantum state preparation algorithms. 
While the overall cost of general state preparation can be expensive, a recent work~\cite{zhang2022quantum} shows that a general $(\log K)$-qubit quantum state can be prepared using a circuit with only $\mathcal{O}(\log K)$ depth and $\mathcal{O}(K)$ many ancilla qubits. 
Alternatively, since amplitudes of the quantum state in Eq.~\eqref{state:U_T0} form a truncated discrete Poisson distribution which is classically integrable, we may alternatively apply the approach in~\cite{grover2002creating} to construct $U_T$ with $\mathcal{O}(\log K)$ gates and depth and $\mathcal{O}(1)$ ancilla qubits. Note that to implement these state preparation algorithms, we also require a pre-fixed unary to binary encoding transformation circuit, which can be achieved by parallel adders via CNOT gates, using also $\mathcal{O}(\log(K))$ depth. As we will see later, the cost of $U_T$ is not dominant in the algorithm implementation compared to other components. 

Next, we introduce a unitary operator $U_g$ with the property
\begin{equation}
\label{ug}
U_g |0\rangle_{b_k} = \sum_{i=1}^M \sqrt{p_i} |i\rangle_{b_k}.
\end{equation}
By treating the qubits in the first environment as control qubits (controlled on the $\ket{1}$ state), we apply a series of controlled-$U_g$ operations to the state in Eq.~\ref{state:U_T0}. This produces the following state:
\begin{equation}\label{st1}
C \sum_{k=0}^K \sqrt{e^{-T} \frac{T^k}{k!}} \, 
|\underbrace{1\cdots1}_k \underbrace{0\cdots0}_{K-k}\rangle_a 
\left( \bigotimes_{l=1}^k \sum_{i=1}^M \sqrt{p_i} |i\rangle_{b_l} \right)
\otimes |0\rangle_{b_{k+1}} \cdots |0\rangle_{b_K}.
\end{equation}

The third type of unitary we require is $U_{F,k}$ (just $U_F$ in Eq.~\ref{uf} with $k$ denoting ancilla systems), defined as
\begin{equation}
U_{F,k} = |0\rangle\langle 0|_{b_k} \otimes I 
+ \sum_{i=1}^M |i\rangle\langle i|_{b_k} \otimes F_i,
\end{equation}
where $I$ and each $F_i$ act on the $n$-qubit Lindbladian system. We sequentially apply $U_{F,k}$ for $k = 1$ to $K$, resulting in the final state:
\begin{align}
&C \sum_{k=0}^K \sqrt{e^{-T} \frac{T^k}{k!}} \, 
|\underbrace{1\cdots1}_k \underbrace{0\cdots0}_{K-k}\rangle_a 
\left( \bigotimes_{l=1}^k \sum_{i=1}^M \sqrt{p_i} |i\rangle_{b_l} \right)
\otimes |0\rangle_{b_{k+1}} \cdots |0\rangle_{b_K} \otimes F_{\vec{i}_k} |\psi_0\rangle \nonumber \\
=& C \sum_{k=0}^K \sqrt{e^{-T} \frac{T^k}{k!}} \, 
|\underbrace{1\cdots1}_k \underbrace{0\cdots0}_{K-k}\rangle_a 
\sum_{\vec{i}_k} |\vec{i}_k\rangle_{b_{1\cdots k}} \otimes |0\rangle_{b_{k+1}} \cdots |0\rangle_{b_K} 
\otimes \sqrt{p_{\vec{i}_k}} F_{\vec{i}_k} |\psi_0\rangle.
\end{align}

By tracing out all the environment systems, we obtain $\tilde{\rho}(t)$, a good approximation to the true density matrix $\rho(t)$. This algorithm effectively implements the following quantum channel:
\begin{equation}
\tilde{\rho}(t) = \mathcal{A}[\rho(0)] := 
C^2 \sum_{k=0}^K \sum_{\vec{i}_k} 
e^{-T} \frac{T^k}{k!} \, p_{\vec{i}_k} \, F_{\vec{i}_k} \rho(0) F_{\vec{i}_k}^\dag.
\end{equation}

For the complexity of this algorithm, since each \( U_{F,k} \) is applied sequentially for \( k = 1 \) to \( K \), while all control-\( U_g \) operations can be implemented in parallel, the circuit depth is determined by the query complexity of \( U_F \). Noting that both \( U_g \) and \( U_F \) have query complexity \( K \), the overall query complexity is \( \mathcal{O}(K) \), and the circuit depth is also \( \mathcal{O}(K) \).

We now present the error analysis, specifically, we prove that $\|e^{\mathcal{L}_dt}-\mathcal{A}\|_\diamond\leq \varepsilon$. 
First, we have
\begin{eqnarray}
e^{\mathcal{L}t}[\rho]-\mathcal{A}[\rho]&=&\left(1-C^2\right)\sum_{k=0}^K\sum_{\vec{i_k}}e^{-T}\frac{T^k}{k!}  p_{\vec{i_k}}F_{\vec{i_k}}\rho F_{\vec{i_k}}^\dag+\sum_{k=K+1}^\infty\sum_{\vec{i_k}} e^{-T}\frac{T^k}{k!}  p_{\vec{i_k}}F_{\vec{i_k}}\otimes F_{\vec{i_k}}^*\nonumber\\
&=&\frac{1-C^2}{C^2}\mathcal{A}[\rho]+\frac{C^2-1}{C^2}\frac{C^2}{C^2-1}\sum_{k=K+1}^\infty\sum_{\vec{i_k}} e^{-T}\frac{T^k}{k!}  p_{\vec{i_k}}F_{\vec{i_k}}\rho F_{\vec{i_k}}^\dag\nonumber\\
&=&\frac{1-C^2}{C^2}\mathcal{A}[\rho]+\frac{C^2-1}{C^2}\mathcal{B}[\rho]
\end{eqnarray}
for some quantum channel $\mathcal{B}[\cdot]$. During the derivation, we used the relation
\begin{eqnarray}
\frac{C^2-1}{C^2}=\sum_{k=K+1}^\infty\sum_{\vec{i_k}} e^{-T}\frac{T^k}{k!}  p_{\vec{i_k}}=\sum_{k=K+1}^\infty e^{-T}\frac{T^k}{k!}.
\end{eqnarray}
Following the definition of the diamond norm, we have
\begin{eqnarray}
\|e^{\mathcal{L}t}-\mathcal{A}\|_\diamond&&=\sup_\sigma \left\|\left(e^{\mathcal{L}t}\otimes \mathcal{I}\right)(\sigma)-\left(\mathcal{A}\otimes \mathcal{I}\right)(\sigma)\right\|_1\nonumber\\
&&=\sup_\sigma \left\|\frac{1-C^2}{C^2}(\mathcal{A}\otimes \mathcal{I})[\sigma]+\frac{C^2-1}{C^2}(\mathcal{B}\otimes \mathcal{I})[\sigma]\right\|_1\nonumber\\&&\leq \sup_\sigma \left\|\frac{1-C^2}{C^2}(\mathcal{A}\otimes \mathcal{I})[\sigma]\right\|_1+\left\|\frac{C^2-1}{C^2}(\mathcal{B}\otimes \mathcal{I})[\sigma]\right\|_1\nonumber\\&&=2\frac{C^2-1}{C^2}=2\sum_{k=K+1}^\infty e^{-T}\frac{T^k}{k!}.
\end{eqnarray}
This upper bound can be further bounded by Taylor's theorem as follows
\begin{eqnarray}
\sum_{k=K+1}^\infty e^{-T}\frac{T^k}{k!}&&=e^{-T}\left(e^{T}-\sum_{k=0}^\infty \frac{T^k}{k!}\right)=e^{-T}\left(\int_0^T\frac{(T-\tau)^K}{K!}e^\tau d\tau\right)\nonumber\\&&\leq\int_0^T\frac{(T-\tau)^K}{K!} d\tau=\frac{T^{K+1}}{(K+1)!}.
\end{eqnarray}
Following the proofs in Lemmas 57-59 of Ref \cite{gilyen2019quantum}, we have
\begin{eqnarray}\label{kep}
\|e^{\mathcal{L}t}-\mathcal{A}\|_\diamond\leq2\frac{T^{K+1}}{(K+1)!}\leq \varepsilon \quad \text{because } \quad K=\mathcal{O}\left(T+\frac{\log(\varepsilon^{-1})}{\log\log(\varepsilon^{-1})}\right).
\end{eqnarray}
This leads to the results in Theorem \ref{mainthe1}.

\textbf{Generalizations:} If we examine the whole proof of Theorem \ref{mainthe1}, we can find that the core is that the Lindbladian should satisfy Eq. \ref{slme}. In fact, Eq. \ref{slme} does not require the jump operators to be exactly unitary. Instead, we can relax the restriction to
\begin{eqnarray}
\sum_{i=1}^M g_i F_i^\dag F_i\propto I,
\end{eqnarray}
which can be applied to non-unitary jump operators (e.g., 1-qubit case: $F_1=|0\rangle\langle 1|  $ and $F_2=|1\rangle\langle 0|$).

\subsection{Proof of Theorem \ref{mainthe2} \label{provethe2}}

We now consider the special case where the operators \( F_i \) take the form
\begin{equation}\label{special}
F_i = \sum_{j=0}^{R-1} |j\rangle\langle j| \otimes P_{j,i},
\end{equation}
where each \( P_{j,i} \) is an \( n \)-qubit Pauli operator from the set \( \{\pm1, \pm i\} \times \{I,X,Y,Z\}^{\otimes n} \). 

The key observation for enabling fast forwarding is that each \( F_{\vec{i_k}} \) in Eq.~\ref{expan} remains in a block-diagonal form:
\begin{equation}
F_{\vec{i_k}} = \sum_{j=0}^{R-1} |j\rangle\langle j| \otimes P_{j,\vec{i_k}},
\end{equation}
where \( P_{j,\vec{i_k}} \) also belongs to the set \( \{\pm1, \pm i\} \times \{I,X,Y,Z\}^{\otimes n} \). 
Therefore, instead of implementing \( P_{j,\vec{i_k}} \) through a sequence of \( k \) individual Pauli operations, we can compute the resulting operator \( P_{j,\vec{i_k}} \) directly using simple arithmetic in the Pauli group. This allows for coherent parallelization and ultimately leads to an exponential reduction in circuit depth.

We use 4 bits to represent a single-qubit Pauli operator in the set \( \{\pm1,\pm i\} \times \{I,X,Y,Z\} \), see Eq. \ref{binary}, where:
\begin{itemize}
    \item The first two bits encode the phase: 00 for \( +1 \), {01} for \( +i \), {10} for \( -1 \), and {11} for \( -i \).
    \item The last two bits encode the Pauli matrix: {00} for \( I \), {01} for \( X \), {10} for \( Y \), and {11} for \( Z \).
\end{itemize}

\begin{eqnarray}\label{pauli}
\renewcommand{\arraystretch}{1.5}
\begin{tabular}{|>{\centering\arraybackslash}p{1cm}
              |>{\centering\arraybackslash}p{1cm}
              |>{\centering\arraybackslash}p{1cm}
              |>{\centering\arraybackslash}p{1cm}
              |>{\centering\arraybackslash}p{1cm}|}
\hline
 & $I$ & $X$ & $Y$ & $Z$ \\
\hline
$I$ & $I$ & $X$ & $Y$ & $Z$ \\
\hline
$X$ & $X$ & $I$ & $iZ$ & $-iY$ \\
\hline
$Y$ & $Y$ & $-iZ$ & $I$ & $iX$ \\
\hline
$Z$ & $Z$ & $iY$ & $-iX$ & $I$ \\
\hline
\end{tabular}
\quad\longrightarrow\quad
\begin{tabular}{|r|r|r|r|r|}
\hline
 & \quad 00$|$00 & \quad 00$|$01 & \quad 00$|$10 & \quad 00$|$11\\
\hline
00$|$00 & 00$|$00 & 00$|$01 & 00$|$10 & 00$|$11\\
\hline
00$|$01 & 00$|$01 & 00$|$00 & 01$|$11 & 11$|$10\\
\hline
00$|$10 & 00$|$10 & 11$|$11 & 00$|$00 & 01$|$01\\
\hline
00$|$11 & 00$|$11 & 01$|$10 & 11$|$01 & 00$|$00\\
\hline
\end{tabular}
\renewcommand{\arraystretch}{1.0}
\end{eqnarray}

The Pauli operator multiplication can be
summarized in Eq. \ref{pauli}.
Thus, a full \( n \)-qubit Pauli operator can be represented as a \( 4n \)-bit string. Pauli operator multiplication can then be expressed via a classical function \( x_3 = f(x_1, x_2) \), where each \( x_i \in \{0,1\}^{4n} \) encodes a Pauli operator. This multiplication function can be implemented coherently via a 12-qubit unitary transformation \( U_P \), acting as:
\begin{equation}
U_P \ket{x_1}\ket{x_2}\ket{0} = \ket{x_1}\ket{x_2}\ket{x_3},
\end{equation}
where \( x_1, x_2 \in \{0,1\}^{4n} \), and \( \ket{x_3} = \ket{f(x_1,x_2)} \) represents the product Pauli operator.

To perform Pauli multiplication for \( K = 2^s \) \( n \)-qubit Pauli operators in binary encoding, we require an \( 8Kn \)-qubit circuit with depth \( \log(K) \). The core idea is to organize the \( 2^s \) Pauli operators (each represented using \( 4n \) bits) into a binary tree structure of pairwise multiplications.

Initially, we divide the \( 2^s \) input \( n \)-qubit Pauli operators into \( 2^{s-1} \) pairs. For each pair, we apply the unitary \( U_P^{\otimes n} \), which performs Pauli multiplication component-wise across \( n \) qubits. These operations can be executed in parallel across all \( 2^{s-1} \) pairs.
At the next level, we take the \( 2^{s-1} \) resulting Pauli operators, group them into \( 2^{s-2} \) pairs, and repeat the multiplication process. Continuing this recursive pairing and multiplication leads to a total circuit depth of \( s = \log(K) \), as each level of the binary tree corresponds to one multiplication round.
The total number of qubits used throughout this process sums to
\begin{equation}
(2^s + 2^{s-1} + \dots + 1) \cdot 4n = (2^{s+1} - 1)\cdot 4n \leq 8Kn.
\end{equation}

Now, consider the operator \( F_{\vec{i_k}} \), which is a product of \( k  \) unitaries \( F_i \). Each \( F_i \) is block-diagonal with \( R \) blocks, and each block is an \( n \)-qubit Pauli operator. Therefore, \( F_{\vec{i_k}} \) also remains block-diagonal with \( R \) blocks, each being a Pauli operator. We can represent the entire operator using \( 4Rn \) qubits.
Applying the above Pauli multiplication circuit to each of the \( R \) blocks in parallel, we can simulate the multiplication of the \( k \) unitary blocks using a total of \( 8RKn \) qubits and circuit depth \( \log(K) \).
We denote the resulting quantum circuit that performs this parallelized Pauli multiplication as \( U_{P,K} \).

We now formalize our fast-forwarding algorithm. Starting from Eq.~\ref{st1}, we append \(4RnK\) ancilla qubits, partitioned into \(K\) groups labeled \(c_1,\ldots,c_K\). We define the binary-encoded version of \(U_{F,k}\), denoted by \(V_{F,k}\) (just $V_F$ in Eq.~\ref{vf} with $k$ denoting ancilla systems), which satisfies
\begin{eqnarray}
V_{F,k}|i\rangle_{b_k}|0\rangle_{c_k} = |i\rangle_{b_k}|x_{F_i}\rangle_{c_k},
\end{eqnarray}
where \(F_0 = I\). The state \(|x_{F_i}\rangle\) represents the binary encoding of the Pauli operator \(F_i\). 

Applying all \(K\) operations \(V_{F,k}\) in parallel results in
\begin{eqnarray}\label{st5}
&&\bigotimes_{r=1}^K V_{F,r} \left( C \sum_{k=0}^K \sqrt{e^{-T} \frac{T^k}{k!}} |\underbrace{1\cdots1}_k \underbrace{0\cdots0}_{K-k}\rangle_a \left( \bigotimes_{l=1}^k \sum_{i=1}^M \sqrt{p_i} |i\rangle_{b_l} \right) \otimes |0\rangle_{b_{k+1}} \cdots |0\rangle_{b_K} \otimes |0\rangle_c \right) \\
&=& \sum_{k=0}^K \sqrt{e^{-T} \frac{T^k}{k!}} |\underbrace{1\cdots1}_k \underbrace{0\cdots0}_{K-k}\rangle_a \left( \bigotimes_{l=1}^k \sum_{i=1}^M \sqrt{p_i} |i\rangle_{b_l} |x_{F_i}\rangle_{c_l} \right) \otimes |0\rangle_{b_{k+1}} \cdots |0\rangle_{b_K} \otimes |0\rangle_{c_{k+1}} \cdots |0\rangle_{c_K} \nonumber \\
&=& \sum_{k=0}^K \sqrt{e^{-T} \frac{T^k}{k!}} |\underbrace{1\cdots1}_k \underbrace{0\cdots0}_{K-k}\rangle_a \sum_{\vec{i_k}} |\vec{i_k}\rangle_{b_{1\ldots k}} \otimes |0\rangle_{b_{k+1}} \cdots |0\rangle_{b_K} \otimes \sqrt{p_{\vec{i_k}}} |x_{F_{i_1}}\rangle_{c_1} \cdots |x_{F_{i_k}}\rangle_{c_k} \otimes |0\rangle_{c_{k+1}} \cdots |0\rangle_{c_K}. \nonumber
\end{eqnarray}

Next, we append another \(4RnK\) ancilla qubits and apply the Pauli multiplication circuit \(U_{P,K}\). Denoting the output register (after discarding intermediate steps) by label \(d\), we obtain
\begin{eqnarray}\label{st2}
\sum_{k=0}^K \sqrt{e^{-T} \frac{T^k}{k!}} |\underbrace{1\cdots1}_k \underbrace{0\cdots0}_{K-k}\rangle_a \sum_{\vec{i_k}} |\vec{i_k}\rangle_{b_{1\ldots k}} \otimes |0\rangle_{b_{k+1}} \cdots |0\rangle_{b_K} \otimes \sqrt{p_{\vec{i_k}}} |x_{F_{\vec{i_k}}}\rangle_d.
\end{eqnarray}
To simulate the action of \(F_{\vec{i_k}}\), we define a transformation \(T_F\) that satisfies
\begin{eqnarray}\label{st3}
T_F |x_{F_{\vec{i_k}}}\rangle |\psi_0\rangle = |x_{F_{\vec{i_k}}}\rangle F_{\vec{i_k}} |\psi_0\rangle,
\end{eqnarray}
where we assume the initial state \(\rho(0) = |\psi_0\rangle\langle\psi_0|\) for simplicity. This transformation \(T_F\) is constructed using a basic 5-qubit building block \(U_{br}\) that acts as
\begin{eqnarray}
U_{br} |x\rangle |\psi\rangle = |x\rangle P_x |\psi\rangle,
\end{eqnarray}
where the classical string \(x\) encodes a single-qubit Pauli operator \(P_x\) as described in Eq.~\ref{pauli}. For example $U_{br} |0110\rangle|\psi\rangle=|0110\rangle(iY)|\psi\rangle$.

Note that the binary string \(|x_{F_{\vec{i_k}}}\rangle_d\) consists of \(R\) blocks of \(4n\) qubits, with the \(j\)th block (denoted \(d_j\)) encoding the Pauli operator \(P_{j,\vec{i_k}}\). The operator \(F_{\vec{i_k}}\) takes the block-diagonal form
\begin{eqnarray}
F_{\vec{i_k}} = \sum_{j=0}^{R-1} |j\rangle\langle j| \otimes P_{j,\vec{i_k}} = \prod_{j=0}^{R-1} \left( \sum_{l\neq j} |l\rangle\langle l| \otimes I + |j\rangle\langle j| \otimes P_{j,\vec{i_k}} \right),
\end{eqnarray}
which is a product of \(R\) mutually commuting controlled-Pauli gates. Thus, \(T_F\) can be written as
\begin{eqnarray}\label{st6}
T_F = \prod_{j=0}^{R-1} \left( \sum_{l\neq j} |l\rangle\langle l| \otimes I + |j\rangle\langle j| \otimes U_{BR,j} \right),
\end{eqnarray}
where \(U_{BR,j}\) consists of \(U_{br}^{\otimes n}\) acting on the \(d_j\) ancilla block and the \(n\)-qubit system, and identity on the rest. The overall circuit depth of \(T_F\) is therefore \(R\). Applying \(T_F\) to Eq.~\ref{st2} yields the final state
\begin{eqnarray}
\sum_{k=0}^K \sqrt{e^{-T} \frac{T^k}{k!}} |\underbrace{1\cdots1}_k \underbrace{0\cdots0}_{K-k}\rangle_a \sum_{\vec{i_k}} |\vec{i_k}\rangle_{b_{1\ldots k}} \otimes |0\rangle_{b_{k+1}} \cdots |0\rangle_{b_K} \otimes \sqrt{p_{\vec{i_k}}} |x_{F_{\vec{i_k}}}\rangle_d F_{\vec{i_k}} |\psi_0\rangle,
\end{eqnarray}
with the reduced state (after tracing out ancilla systems) approximating \(\tilde{\rho}(t)\).

Regarding the complexity of this algorithm, as established in Eq.~\ref{kep}, the number of terms \(K\) satisfies
\begin{equation}
K = \mathcal{O}\left( t + \frac{\log(1/\varepsilon)}{\log\log(1/\varepsilon)} \right).
\end{equation}
From the algorithm structure:
\begin{itemize}
\item The query complexities of both \(U_g\) and \(V_F\) are \(K\), as each appears \(K\) times.
\item Since \(U_g\) in Eq. \ref{ug} and \(V_{F}\) in Eq. \ref{st5}  can be applied simultaneously, the overall circuit depth is dominated by:

(i) The Pauli multiplication circuit \(U_{P,K}\), which has depth \(\log(K)\),

(ii) The application of \(T_F\), which has depth \(R\).

\item The query complexity of $U_T$ is $\mathcal{O}(1)$, and constructing each $U_T$ requires $\mathcal{O}(\log K)$ depth and $\mathcal{O}(K)$ additional ancilla qubits. 
Both costs are asymptotically less dominant compared to those of $U_g$ and $V_F$. 
\end{itemize}
Combining all contributions, we obtain the complexity stated in Theorem~\ref{mainthe2}.

\section{Estimating Gibbs coherence amplitude\label{esgca}}
\label{sec: app c}

\subsection{Non-diagonal density matrix encoding}

Here, we introduce the \emph{non-diagonal density matrix encoding} (NDME) technique. The key idea is that, given an arbitrary $n$-qubit matrix, we can encode it into a non-diagonal block of an $(n+1)$-qubit density matrix.

\begin{definition}[Non-diagonal density matrix encoding (NDME) \cite{shang2024estimating,shang2024design}]
Given an $(n+1)$-qubit density matrix $\rho_M$ and an $n$-qubit matrix $M$, if $\rho_M$ has the form
\begin{equation}
\rho_M = \begin{pmatrix}
\cdot & \gamma M \\
\gamma M^\dag & \cdot
\end{pmatrix},
\end{equation}
with $\gamma \geq 0$, we say that $\rho_M$ is a $\gamma$-NDME of $M$.
\end{definition}

NDME allows us to bypass the Hermitian and positive semi-definite constraints of density matrices and instead use density matrix elements to encode general matrix information.

\subsection{Pauli encoding: $I$ as $|0\rangle$ and $X$ as $|1\rangle$}

Here, we formalize the idea of \emph{Pauli encoding}, in which the Pauli operators $I$ and $X$ are used in the non-diagonal blocks of density matrices to encode the computational basis states $|0\rangle$ and $|1\rangle$. Our discussion primarily uses the vectorization picture for illustration. 

Through the processes of vectorization and matrixization, we arrive at the following \emph{Pauli–Bell correspondence}:
\begin{eqnarray}\label{PBC}
I=\begin{pmatrix}1& 0\\0&1\end{pmatrix} 
&\xrightarrow{\mathcal{V}}& \sqrt{2}|\Phi^+\rangle = |0\rangle|0\rangle + |1\rangle|1\rangle, \nonumber\\
X=\begin{pmatrix}0&1\\1&0\end{pmatrix} 
&\xrightarrow{\mathcal{V}}& \sqrt{2}|\Psi^+\rangle = |0\rangle|1\rangle + |1\rangle|0\rangle, \nonumber\\
Y=\begin{pmatrix}0&-i\\i&0\end{pmatrix} 
&\xrightarrow{\mathcal{V}}& -i\sqrt{2}|\Psi^-\rangle = -i|0\rangle|1\rangle + i|1\rangle|0\rangle, \nonumber\\
Z=\begin{pmatrix}1& 0\\0&-1\end{pmatrix} 
&\xrightarrow{\mathcal{V}}& \sqrt{2}|\Phi^-\rangle = |0\rangle|0\rangle - |1\rangle|1\rangle.
\end{eqnarray}

For $I$ and $X$ in Eq. \ref{PBC}, we have the observation
\begin{eqnarray}
U_B|\Phi^+\rangle=|0\rangle|0\rangle,
\nonumber\\
U_B|\Psi^+\rangle=|0\rangle|1\rangle,
\end{eqnarray}
where $U_B$ is the Bell circuit
\begin{equation}
U_B=\Qcircuit @C=1em @R=.7em {
&\qw &\ctrl{1}  & \gate{H} &\qw \\
&\qw &\targ  &  \qw &\qw}=(H\otimes I)\text{CNOT}.
\end{equation}
Therefore, the purpose of treating $I$ as $|0\rangle$ and $X$ as $|1\rangle$ can be formalized by the mapping $\mathcal{PQC}[\cdot]$:
\begin{eqnarray}\label{pqcmapping}
\mathcal{PQC}[I] &=& (\langle 0| \otimes I) U_B \mathcal{V}[I] = \sqrt{2} |0\rangle, \nonumber \\
\mathcal{PQC}[X] &=& (\langle 0| \otimes I) U_B \mathcal{V}[X] = \sqrt{2} |1\rangle.
\end{eqnarray}
Importantly, the mapping $\mathcal{PQC}[\cdot]$ does not carry any operational meaning; it is introduced solely to provide a formal understanding of Pauli encoding. Note that the mapping $\mathcal{PQC}[\cdot]$ can be straightforwardly generalized to the $n$-qubit case.

Through the mappings $\mathcal{PQC}[\cdot]$ and NDME, we are able to encode an arbitrary $n$-qubit pure state into a non-diagonal block of a density matrix. Given an $n$-qubit pure state $|S\rangle$, we can express it in the computational basis as
\begin{equation}
|S\rangle = \sum_i c_i |i\rangle,
\end{equation}
with $\sum_i |c_i|^2 = 1$. 
Since, in PQC, each basis state $|i\rangle$ can be encoded as an $n$-qubit Pauli operator $P_i$ containing only $I$ and $X$, we have
\begin{equation}
\mathcal{PQC}\left[S = 2^{-{n}/{2}} \sum_i c_i P_i \right] = |S\rangle,
\end{equation}
where the factor $2^{-n/2}$ arises from the Frobenius norm $\|P_i\|_F = 2^{n/2}$. 
Thus, a $\gamma$-NDME $\rho_S$ of $S$ is a density matrix encoding the information of $|S\rangle$ as
\begin{equation}\label{ndmes}
\rho_S = \begin{pmatrix}
\cdot & \gamma S \\
\gamma S^\dagger & \cdot
\end{pmatrix}.
\end{equation}
In the following, we will sometimes abuse notation and say that $\rho_S$ is a $\gamma$-NDME of $|S\rangle$. 
As an example, the state $\rho_S = |+\rangle \langle +|^{\otimes (n+1)}$ is a $1/2$-NDME of $|+\rangle \langle +|^{\otimes n}$ and effectively encodes $|S\rangle = |+\rangle^{\otimes n}$, since $I + X$ corresponds to $|0\rangle + |1\rangle$ under the mapping.

The NDME encoding factor $\gamma$ has an upper bound determined by properties of the state $|S\rangle$.

\begin{lemma}[Upper bound of $\gamma$]\label{ubg}
For any state $\ket{S}$, its $\gamma$-NDME satisfies
\begin{equation}
\gamma \leq \gamma_S := \frac{1}{2 \sum_i \big|\langle i | \mathcal{H}^{\otimes n} | S \rangle \big|} = \frac{1}{2 \| \mathcal{H}^{\otimes n} | S \rangle \|_1},
\end{equation}
where $\mathcal{H}$ is the Hadamard gate.
\end{lemma}

\begin{proof}
Note that
\begin{align}
\mathcal{PQC}[\mathcal{H} |0\rangle \langle 0| \mathcal{H}] &= \mathcal{PQC}\left[\frac{I+X}{2}\right] = \mathcal{H} |0\rangle, \\
\mathcal{PQC}[\mathcal{H} |1\rangle \langle 1| \mathcal{H}] &= \mathcal{PQC}\left[\frac{I - X}{2}\right] = \mathcal{H} |1\rangle,
\end{align}
which implies
\begin{equation}
\mathcal{PQC}\left[\mathcal{H}^{\otimes n} |i\rangle \langle i| \mathcal{H}^{\otimes n}\right] = \mathcal{H}^{\otimes n} |i\rangle.
\end{equation}
This leads to
\begin{equation}
\sum_i \big| \langle i | \mathcal{H}^{\otimes n} | S \rangle \big| = \sum_i \left| \Tr \left( \mathcal{H}^{\otimes n} |i\rangle \langle i| \mathcal{H}^{\otimes n} S \right) \right|.
\end{equation}

Consider the transformation
\begin{equation}
(I \otimes \mathcal{H}^{\otimes n}) \rho_S (I \otimes \mathcal{H}^{\otimes n}) = \begin{pmatrix}
R_0 & \gamma \Sigma_S \\
\gamma \Sigma_S^* & R_1
\end{pmatrix},
\end{equation}
where $R_0$ and $R_1$ are unnormalized density matrices, and $\Sigma_S = \mathcal{H}^{\otimes n} S \mathcal{H}^{\otimes n}.$
Since $\rho_S$ is a density operator, we have $\Tr(R_0 + R_1) = 1$. 
Moreover, because $S$ is a linear combination of only the Pauli operators $I$ and $X$, $\Sigma_S=\mathcal{H}^{\otimes n}S \mathcal{H}^{\otimes n}=\sum_i \chi_i |i\rangle\langle i|$ is diagonal.
By positive semi-definiteness of density matrices, it follows that
\begin{equation}
\gamma \sum_i |\chi_i| \leq \gamma \sum_i \sqrt{R_{0,i} R_{1,i}} \leq \sum_i \frac{R_{0,i} + R_{1,i}}{2} = \frac{1}{2},
\end{equation}
where $R_{0,i}$ and $R_{1,i}$ denote the diagonal elements of $R_0$ and $R_1$, respectively.
We obtain the claimed upper bound by noting that $\chi_i = \Tr ( \mathcal{H}^{\otimes n} |i\rangle \langle i| \mathcal{H}^{\otimes n} S )$.
\end{proof}

As examples, when $|S\rangle$ is a computational basis state, its upper bound $\gamma_S = 2^{-n/2-1}$
is the smallest. In contrast, when $|S\rangle$ is in the $+/-$ basis, its upper bound $\gamma_S =1/2$ is the largest.

\subsection{Channel block encoding}

When we have $\rho_S$ at hand, what kinds of operations can we perform to manipulate $S$ and, therefore, the state $|S\rangle$? This question can be answered by a new technique called \emph{channel block encoding} (CBE) \cite{shang2024estimating,shang2024design}. 

CBE considers a quantum channel $\mathcal{C}[\cdot]$ acting on the input $\rho_S$ (defined in Eq.~\ref{ndmes}) as
\begin{equation}\label{channel}
\mathcal{C}[\rho_{S}] = \sum_i
\begin{pmatrix}
K_i & 0 \\
0 & L_i
\end{pmatrix}
\rho_{S}
\begin{pmatrix}
K_i^\dag & 0 \\
0 & L_i^\dag
\end{pmatrix}
= \begin{pmatrix}
\cdot & \gamma \sum_i K_i S L_i^\dag \\
\gamma \sum_i L_i S^\dag K_i^\dag & \cdot
\end{pmatrix}.
\end{equation}
To ensure complete positivity and trace preservation, we require $\sum_i K_i^\dag K_i = \sum_i L_i^\dag L_i = I.$
Focusing on the upper-right block, under the mapping $\mathcal{PQC}[\cdot]$, the channel $\mathcal{C}[\cdot]$ implements
\begin{align}
\gamma |S\rangle \xrightarrow{\mathcal{C}} \gamma \, \mathcal{PQC}\left[\sum_i K_i S L_i^\dag\right]
&= \gamma (\langle 0|^{\otimes n} \otimes I) U_B^{\otimes n} \mathcal{V}\left[\sum_i K_i S L_i^\dag\right] \nonumber \\
&= \gamma (\langle 0|^{\otimes n} \otimes I) U_B^{\otimes n} \left(\sum_i K_i \otimes L_i^*\right) \mathcal{V}[S] \nonumber \\
&= \gamma (\langle 0|^{\otimes n} \otimes I) U_B^{\otimes n} \left(\sum_i K_i \otimes L_i^*\right) U_B^{\dag \otimes n} |0\rangle^{\otimes n} |S\rangle.
\end{align}
Thus, by carefully choosing the sets $\{K_i\}$ and $\{L_i\}$, we can realize any desired operation on $|S\rangle$.

\begin{definition}[Channel block encoding (CBE)]\label{cbepqc}
Given an $n$-qubit operator $Q$, a $(n+1)$-qubit quantum channel $\mathcal{C}[\cdot]$ of the form Eq.~\ref{channel} is called an \emph{$\eta$-CBE of $Q$ in PQC} if
\begin{equation}\label{cbec}
(\langle 0|^{\otimes n} \otimes I_n) U_B^{\otimes n} \left(\sum_i K_i \otimes L_i^*\right) U_B^{\dag \otimes n} (|0\rangle^{\otimes n} \otimes I_n) = \eta Q,
\end{equation}
for some $\eta > 0$.
\end{definition}

Given two quantum channels $\mathcal{C}_1[\cdot]$ with $\{K_{1,i}, L_{1,i}\}$ and $\mathcal{C}_2[\cdot]$ with $\{K_{2,i}, L_{2,i}\}$, which are $\eta_1$-CBE of $Q_1$ and $\eta_2$-CBE of $Q_2$ respectively, the composite channel $\mathcal{C}_2[\mathcal{C}_1[\cdot]]$ is not necessarily a CBE of the product $Q_2 Q_1$. To ensure this, the following stronger condition must hold:
\begin{align}\label{strongCBE}
&(\langle 0|^{\otimes n} \otimes I_n) U_B^{\otimes n} \left(\sum_i K_{2,i} \otimes L_{2,i}^*\right) U_B^{\dag \otimes n} (|0\rangle^{\otimes n} \otimes I_n) \nonumber \\
&\quad \times (\langle 0|^{\otimes n} \otimes I_n) U_B^{\otimes n} \left(\sum_i K_{1,i} \otimes L_{1,i}^*\right) U_B^{\dag \otimes n} (|0\rangle^{\otimes n} \otimes I_n) \nonumber \\
=& \eta_1 \eta_2 Q_2 Q_1 \nonumber \\
=& (\langle 0|^{\otimes n} \otimes I_n) U_B^{\otimes n} \left(\sum_i K_{2,i} \otimes L_{2,i}^*\right) \left(\sum_j K_{1,j} \otimes L_{1,j}^*\right) U_B^{\dag \otimes n} (|0\rangle^{\otimes n} \otimes I_n).
\end{align}
We call such CBEs {strong-CBEs}. In short, strong-CBEs have the associative property.

\subsection{Optimal channel block encoding for elementary quantum gates\label{gatecbe}}
For any unitary operator $Q = U$, its CBE factor $\eta$ has an upper bound $\eta_U$.

\begin{lemma}[Upper bound of $\eta$ for unitary operators]\label{ube}
\begin{equation}
\eta \leq \eta_U = \inf_{|S\rangle} \frac{\|\mathcal{H}^{\otimes n}|S\rangle\|_1}{\|\mathcal{H}^{\otimes n} U |S\rangle\|_1}.
\end{equation}
\end{lemma}

\begin{proof}
Both $|S\rangle$ and $U|S\rangle$ have NDME upper bounds $\gamma_S$ and $\gamma_{U,S}$, respectively, by Lemma \ref{ubg}. Since applying the CBE of $U$ to the NDME of $|S\rangle$ yields an NDME of $U|S\rangle$, we must have
\begin{equation}
\eta \gamma_S \leq \gamma_{U,S},
\end{equation}
which holds for arbitrary $|S\rangle$. Taking the infimum over all $|S\rangle$ gives the claimed bound.
\end{proof}

We now provide concrete optimal strong-CBE constructions for $Q$ being elementary unitary quantum gates listed in Table \ref{tbl1}. The optimality means these constructions achieve their respective upper bound $\eta = \eta_U$.

\begin{table}[htbp]
\centering  
\begin{tabular}{|c|c|c|c|}  
\hline 
Gate &  Kraus &  Gate &Kraus  \\  
\hline
&  & &\\
$X$ & $\begin{pmatrix}
I & 0\\
0& X
\end{pmatrix}.$ & $\mathcal{H}$ & $\frac{1}{2}\begin{pmatrix}
I & 0\\
0& X
\end{pmatrix},\frac{1}{2}\begin{pmatrix}
Z & 0\\
0& Z
\end{pmatrix},\frac{1}{2}\begin{pmatrix}
X & 0\\
0& I
\end{pmatrix},\frac{1}{2}\begin{pmatrix}
Y & 0\\
0& Y
\end{pmatrix}.$\\  
& & &\\
\hline
& & &\\
$Y$ & $\begin{pmatrix}
Z & 0\\
0& -Y
\end{pmatrix}.$ & $\mathcal{H}\mathcal{S}\mathcal{H}$ & $\frac{1}{\sqrt{2}}\begin{pmatrix}
I & 0\\
0& \frac{1}{2}\begin{pmatrix}
1+i & 1-i\\
1-i& 1+i
\end{pmatrix}
\end{pmatrix},\frac{1}{\sqrt{2}}\begin{pmatrix}
X & 0\\
0& \frac{1}{2}\begin{pmatrix}
1-i & 1+i\\
1+i& 1-i
\end{pmatrix}
\end{pmatrix}. $ \\  
& & & \\
\hline
& & & \\
$Z$ & $\begin{pmatrix}
Z & 0\\
0& Z
\end{pmatrix}.$&  $\mathcal{H}T\mathcal{H}$ &$\frac{1}{\sqrt{2}}\begin{pmatrix}
I & 0\\
0&\frac{1}{2}\begin{pmatrix}
1+e^{i\pi/4} & 1-e^{i\pi/4}\\
1-e^{i\pi/4}& 1+e^{i\pi/4}
\end{pmatrix}
\end{pmatrix},\frac{1}{\sqrt{2}}\begin{pmatrix}
X & 0\\
0& \frac{1}{2}\begin{pmatrix}
1-e^{i\pi/4} & 1+e^{i\pi/4}\\
1+e^{i\pi/4}& 1-e^{i\pi/4}
\end{pmatrix}
\end{pmatrix}. $\\  
&  & &\\
\hline
\multicolumn{2}{|c}{Gate} & \multicolumn{2}{|c|}{Kraus}\\
\hline
\multicolumn{2}{|c}{}& \multicolumn{2}{|c|}{}\\
\multicolumn{2}{|c}{$(\mathcal{H}\otimes \mathcal{H})CNOT (\mathcal{H}\otimes \mathcal{H})$} &\multicolumn{2}{|c|}{\makecell{$\frac{1}{2\sqrt{2}}\begin{pmatrix}
I\otimes I & 0\\
0& I\otimes I
\end{pmatrix},\frac{1}{2\sqrt{2}}\begin{pmatrix}
I\otimes X & 0\\
0& I\otimes X
\end{pmatrix},\frac{1}{2\sqrt{2}}\begin{pmatrix}
Z\otimes I & 0\\
0& X\otimes I
\end{pmatrix},$\\ $\frac{1}{2\sqrt{2}}\begin{pmatrix}
Z\otimes X & 0\\
0& X\otimes X
\end{pmatrix},\frac{1}{2\sqrt{2}}\begin{pmatrix}
X\otimes I & 0\\
0& Z\otimes I
\end{pmatrix},\frac{1}{2\sqrt{2}}\begin{pmatrix}
X\otimes X & 0\\
0& Z\otimes X
\end{pmatrix},$\\$\frac{1}{2\sqrt{2}}\begin{pmatrix}
Y\otimes I & 0\\
0& -Y\otimes I
\end{pmatrix},\frac{1}{2\sqrt{2}}\begin{pmatrix}
Y\otimes X & 0\\
0& -Y\otimes X
\end{pmatrix}.$}  }  \\
\multicolumn{2}{|c}{}& \multicolumn{2}{|c|}{}\\
\hline
\end{tabular}
\caption{Optimal CBE constructions for elementary quantum gates. 
Each gate is decomposed into CBE Kraus operators of the form $\begin{pmatrix} K_i & 0 \\ 0 & L_i \end{pmatrix}$. 
All constructions reach their upper bounds. The bound for $\mathcal{H}$ is $1/\sqrt{2}$, while all others attain $\eta_U = 1$.
\label{tbl1}}
\end{table}

Note that $\eta_U = 1$ implies
\[
\sum_i \big|\langle i | \mathcal{H}^{\otimes n} U | S \rangle \big| = \sum_i \big| \langle i | \mathcal{H}^{\otimes n} | S \rangle \big|
\]
for all $|S\rangle$, which holds if and only if $U$ is a composition of permutation and phase gates in the $+/-$ basis. Equivalently, $\mathcal{H}^{\otimes n} U \mathcal{H}^{\otimes n}$ is a composition of permutation and phase gates in the computational basis.
This is true for 
\[
U \in \{ X, Y, Z, \mathcal{H} \mathcal{S} \mathcal{H}, \mathcal{H} T \mathcal{H}, (\mathcal{H} \otimes \mathcal{H}) \mathrm{CNOT} (\mathcal{H} \otimes \mathcal{H}) \},
\]
where $\mathcal{S} = |0\rangle\langle 0| + i |1\rangle \langle 1|$ (not to be confused with $S$ in Eq.~\ref{ndmes}).
On the other hand, for the Hadamard gate $\mathcal{H}$, since $\mathcal{H}|0\rangle = |+\rangle,$
where $|+\rangle$ has $\gamma \leq 1/2$ and $|0\rangle$ has $\gamma \leq 1/(2\sqrt{2})$, the upper bound for $\eta$ of $\mathcal{H}$ is
\[
\eta \leq \frac{1/2}{1/(2\sqrt{2})} = \frac{1}{\sqrt{2}}.
\]

The strongest CBE constructions in Table \ref{tbl1} satisfy the following identities:
\begin{align}
\mathcal{H}: \quad
U_B \sum_i K_i \otimes L_i^* U_B^\dag &= |0\rangle \langle 0| \otimes \mathcal{H}, \\
\mathcal{H} \mathcal{S} \mathcal{H}: \quad
U_B \sum_i K_i \otimes L_i^* U_B^\dag &= |0\rangle \langle 0| \otimes \mathcal{H} \mathcal{S} \mathcal{H}, \\
\mathcal{H} T \mathcal{H}: \quad
U_B \sum_i K_i \otimes L_i^* U_B^\dag &= |0\rangle \langle 0| \otimes \mathcal{H} T \mathcal{H}, \\
(\mathcal{H} \otimes \mathcal{H}) \mathrm{CNOT} (\mathcal{H} \otimes \mathcal{H}): \quad
U_B^{\otimes 2} \sum_i K_i \otimes L_i^* U_B^{\dag \otimes 2} &= |00\rangle \langle 00| \otimes (\mathcal{H} \otimes \mathcal{H}) \mathrm{CNOT} (\mathcal{H} \otimes \mathcal{H}),
\end{align}
and
\begin{align}
X: \quad U_B (I \otimes X) U_B^\dag &= I \otimes X, \label{ccx}\\
Y: \quad U_B (Z \otimes Y) U_B^\dag &= I \otimes Y, \label{ccy}\\
Z: \quad U_B (Z \otimes Z) U_B^\dag &= I \otimes Z. \label{ccz}
\end{align}
All these satisfy the strong-CBE condition Eq.~\ref{strongCBE}.

Since these CBE channels will be implemented on quantum computers, their corresponding purification unitaries are important. Given a quantum channel $\mathcal{C}[\rho] = \sum_i A_i \rho A_i^\dag$ with $\sum_i A_i^\dag A_i = I$, the standard Stinespring dilation \cite{preskill1998lecture} gives a purification unitary $U_C$ satisfying
\begin{equation}
U_C |0\rangle |\psi\rangle = \sum_i |i\rangle A_i |\psi\rangle.
\end{equation}
Thus, the purification unitary for each CBE construction in Table \ref{tbl1} can be constructed accordingly. The number of supporting qubits for these unitaries equals the number of target gate qubits plus the logarithm of the number of Kraus operators. For instance, the purification unitary of the CBE of $CNOT$ is a 5-qubit unitary.

\subsection{Exponential amplitude amplification}

With $\rho_S$ at hand, we now show how to estimate its amplitudes. For computational basis amplitudes of the form $\langle j |S\rangle$, where $|j\rangle$ is a computational basis state, the Pauli expectation values of $\rho_S$ with respect to $X \otimes P_j$ and $Y \otimes P_j$ encode the information of the amplitude $\langle j|S\rangle$ through the relation
\begin{eqnarray}\label{amp1}
\Tr((X\otimes P_j )\rho_S) &=& \gamma\Tr(P_j S)+\gamma\Tr(P_j S^\dag) \\
&=& \gamma\left(2^{-n/2} c_j \Tr(P_j^2)+2^{-n/2} c_j^* \Tr(P_j^2)\right) \nonumber\\
&=& 2^{n/2+1}\gamma\re[\langle j |S\rangle], \nonumber\\
\Tr((Y\otimes P_j)\rho_S) &=& \gamma\Tr(P_j S)+\gamma\Tr(P_j S^\dag)\label{amp2} \\
&=& \gamma\left(-i2^{-n/2} c_j \Tr(P_j^2)+i2^{-n/2} c_j^* \Tr(P_j^2)\right) \nonumber\\
&=& 2^{n/2+1}\gamma\im[\langle j |S\rangle], \nonumber
\end{eqnarray}
where we used the decomposition $S=2^{-n/2}\sum_j c_j P_j$. If we are further given the purification $|P_S\rangle$ of $\rho_S$, the following lemma provides a reformulation of Eqs.~\ref{amp1} and~\ref{amp2}.

\begin{lemma}\label{amp3}
Given $|P_S\rangle$ as the purification of $\rho_S$, $\langle j|S\rangle$ can be estimated as
\begin{eqnarray}
2|\langle 0|\langle P_S|U_{r,P_j}\otimes I_e |0\rangle|P_S\rangle| &=& 1+2^{n/2+1}\gamma \re[\langle j |S\rangle], \\
2|\langle 0|\langle P_S|U_{i,P_j}\otimes I_e |0\rangle|P_S\rangle| &=& 1+2^{n/2+1}\gamma \im[\langle j |S\rangle],
\end{eqnarray}
where $I_e$ denotes the identity operator on the environment register, and $U_{r,P_j}$ and $U_{i,P_j}$ are block encodings of the form
\begin{eqnarray}
U_{r,P_j}=\begin{pmatrix}
\frac{1}{2} (X\otimes P_j+I\otimes I_n)	&	\cdot	\\
\cdot	&	\cdot	\\
\end{pmatrix},\quad
U_{i,P_j}=\begin{pmatrix}
\frac{1}{2} (Y\otimes P_j+I\otimes I_n) &	\cdot	\\
\cdot	&	\cdot	\\
\end{pmatrix}.
\end{eqnarray}
\end{lemma}

\begin{proof}
We have
\begin{eqnarray}
\langle 0|\langle P_S|U_{r,P_j}\otimes I_e |0\rangle|P_S\rangle
&=& \Tr(|0\rangle\langle 0|\otimes |P_S\rangle\langle P_S| \cdot U_{r,P_j}\otimes I_e) \nonumber\\
&=& \Tr(|0\rangle\langle 0|\otimes \Tr_e(|P_S\rangle\langle P_S|) U_{r,P_j}) \nonumber\\
&=& \frac{1}{2} \Tr(\rho_S (X\otimes P_j + I\otimes I_n)) \nonumber\\
&=& \frac{1}{2} + \frac{2^{n/2+1} \gamma \re[\langle j |S\rangle]}{2},
\end{eqnarray}
where in the last line we apply Eq.~\ref{amp1}. Since $|\Tr(\rho_S X\otimes P_j)| \leq 1$, taking the absolute value does not change the result. The proof for the imaginary part is similar.
\end{proof}

Through Lemma~\ref{amp3}, we can now use amplitude estimation to probe the information of $\langle j |S\rangle$. The complexity of estimating $\langle j|S\rangle$ is summarized in the following lemma.

\begin{lemma}\label{thea}
Using the notation in Lemma~\ref{amp3}, there is a quantum algorithm that returns estimates $\mu_r$ of $\re[\langle j |S\rangle]$ and $\mu_i$ of $\im[\langle j |S\rangle]$, each up to an additive error $\epsilon$, with success probability at least $1-\delta$. Moreover, the query complexity on the preparation circuits of $|P_S\rangle$ and $U_{r(i),P_j}$ (and their inverses) is 
\begin{eqnarray}
\mathcal{O}\left(2^{-n/2}\gamma^{-1}\epsilon^{-1}\log\left(\delta^{-1}\right)\right).
\end{eqnarray}
\end{lemma}

\begin{proof}
From Lemma~\ref{amp3}, the requirement $|\mu_{r(i)} - \re(\im)[\langle j |S\rangle]| \leq \epsilon$ corresponds to estimating $\langle 0|\langle P_S|U_{r(i),P_j}\otimes I_e |0\rangle|P_S\rangle$ to within an additive error of approximately $2^{n/2} \gamma \epsilon$. By applying Lemma~\ref{oe}, the result follows.
\end{proof}

\subsection{Proof of Theorem \ref{mainthe3}\label{proofthe3}}

Our goal is to estimate the GCA
$
\langle \psi_1|e^{-\beta (H+I)}|\psi_2\rangle.
$
Since $|\psi_1\rangle = U_1|+\rangle^{\otimes n}$ and $|\psi_2\rangle = U_2|0\rangle^{\otimes n}$, we have
\begin{align}
\langle \psi_1|e^{-\beta (H+I)}|\psi_2\rangle
&= \langle +|^{\otimes n} U_1^\dag e^{-\beta (H+I)} U_2 |0\rangle^{\otimes n} \nonumber \\
&= \langle 0|^{\otimes n} \mathcal{H}^{\otimes n} U_1^\dag e^{-\beta (H+I)} U_2 \mathcal{H}^{\otimes n} |+\rangle^{\otimes n} \nonumber \\
&= \langle 0|^{\otimes n} \left( \mathcal{H}^{\otimes n} U_1^\dag \mathcal{H}^{\otimes n} \right)
\left( \mathcal{H}^{\otimes n} e^{-\beta (H+I)} \mathcal{H}^{\otimes n} \right)
\left( \mathcal{H}^{\otimes n} U_2 \mathcal{H}^{\otimes n} \right) |+\rangle^{\otimes n}, \label{oorr}
\end{align}
where $\mathcal{H}$ denotes the Hadamard gate.

The idea for estimating the GCA is to use Pauli encoding and to construct
\[
\mathcal{H}^{\otimes n} U_1^\dag e^{-\beta (H+I)} U_2 \mathcal{H}^{\otimes n}
\]
through CBE.
Suppose $U_2 = G_{2,D_2} \cdots G_{2,1}$ consists of $D_2$ layers, where each $G_{2,i}$ is composed of gates from the set $\{\mathcal{H}, \mathcal{S}, T, \mathrm{CNOT}\}$. Then
\[
\mathcal{H}^{\otimes n} U_2 \mathcal{H}^{\otimes n} =
\left( \mathcal{H}^{\otimes n} G_{2,D_2} \mathcal{H}^{\otimes n} \right)
\cdots
\left( \mathcal{H}^{\otimes n} G_{2,1} \mathcal{H}^{\otimes n} \right).
\]
For each $\mathcal{H}^{\otimes n} G_{2,i} \mathcal{H}^{\otimes n}$, we can use the CBE constructions of
\[
\left\{ \mathcal{H},\, \mathcal{H} \mathcal{S} \mathcal{H},\, \mathcal{H} T \mathcal{H},\, (\mathcal{H} \otimes \mathcal{H})\, \mathrm{CNOT}\, (\mathcal{H} \otimes \mathcal{H}) \right\}
\]
as listed in Table~\ref{tbl1}, to form a 1-depth quantum channel that is a CBE of $\mathcal{H}^{\otimes n} G_{2,i} \mathcal{H}^{\otimes n}$, since the gate locality is preserved under these CBE constructions.
All these constructions are strong CBEs; therefore, the composition of these channels over depth $D_2$, denoted as $\mathcal{C}_{u_2}[\cdot]$, forms a CBE of $\mathcal{H}^{\otimes n} U_2 \mathcal{H}^{\otimes n}$. Since only the Hadamard gate's CBE contributes a factor $\eta = 1/\sqrt{2}$ while the others contribute $\eta = 1$, the overall amplitude factor of the resulting CBE is $1/2^{n_{h_2}/2}$, where $n_{h_2}$ is the total number of Hadamard gates in the circuit $U_2$.
Following the same procedure, we can construct a CBE of $\mathcal{H}^{\otimes n} U_1^\dag \mathcal{H}^{\otimes n}$, with amplitude factor $\eta = 1/2^{n_{h_1}/2}$, and denote the corresponding channel as $\mathcal{C}_{u_1^\dag}[\cdot]$. The depth of $\mathcal{C}_{u_1^\dag}[\cdot]$ is $D_1$.

Now we show how to combine CBE with Lindbladian to encode 
\begin{eqnarray}
\mathcal{H}^{\otimes n} e^{-\beta (H+I)}\mathcal{H}^{\otimes n}=e^{-\beta(H'+I)},
\end{eqnarray}
where
\begin{eqnarray}
H' :=\mathcal{H}^{\otimes n} H \mathcal{H}^{\otimes n}=\sum_{i=1}^M\lambda_i \mathcal{H}^{\otimes n}Q_i\mathcal{H}^{\otimes n}=\sum_{i=1}^M\lambda_i Q'_i.
\end{eqnarray}
Here $Q_i' =  \mathcal{H}^{\otimes n}Q_i\mathcal{H}^{\otimes n}$ are also Pauli strings.
Consider a Lindbladian with a similar form to Eq. \ref{slme}
\begin{eqnarray}
\frac{d\rho(t)}{dt}=\mathcal{L}_H[\rho(t)]=\sum_{i=1}^M\lambda_i\begin{pmatrix}
P_{0,i} & 0\\
0& P_{1,i}
\end{pmatrix}\rho(t) \begin{pmatrix}
P_{0,i} & 0\\
0& P_{1,i}
\end{pmatrix}-\rho(t),
\end{eqnarray}
with $P_{0,i}, P_{1,i}\in \{ \pm1\} \times\{I,X,Y,Z\}^{\otimes n}$. We use the fact $\sum_{i=1}^M\lambda_i=1$. Focusing on the dynamics of the upper-right block $\rho(t)_{01}$ of $\rho(t)$, we have
\begin{eqnarray}
\frac{d\rho(t)_{01}}{dt}=\sum_{i=1}^M\lambda_i P_{0,i}\rho(t)_{01} P_{1,i}- \rho(t)_{01}.
\end{eqnarray}
In the vectorization picture,  the solution is
\begin{eqnarray}
\ket{\rho(t)_{01}\rangle}=\exp\left(\sum_{i=1}^M\lambda_i t P_{0,i}\otimes P_{1,i}^*-t\right)\ket{\rho(0)_{01}\rangle}.
\end{eqnarray}
Recalling the definition of CBE (Definition~\ref{cbepqc}), to make this dynamical process a CBE of $e^{-\beta(H'+I)}$, we require:
\begin{eqnarray}
&&\left(\langle 0|^{\otimes n} \otimes I_n\right) 
U_B^{\otimes n}  
\exp\left( \sum_{i=1}^M \lambda_i t\, P_{0,i} \otimes P_{1,i}^* - t \right) 
U_B^{\dag \otimes n} (|0\rangle^{\otimes n} \otimes I_n) \nonumber \\
&=& \left(\langle 0|^{\otimes n} \otimes I_n\right)  
\exp\left( \sum_{i=1}^M \lambda_i t\, U_B^{\otimes n} (P_{0,i} \otimes P_{1,i}^*) U_B^{\dag \otimes n} - t \right) 
(|0\rangle^{\otimes n} \otimes I_n) \nonumber \\
&=& e^{-\beta(H'+I)}.
\end{eqnarray}

Consider the CBE constructions of Pauli gates listed in Table~\ref{tbl1}, together with the properties given in Eqs.~\ref{ccx}–\ref{ccz}. Based on these constructions, for each $Q_i' \in \{\pm1\} \times \{I, X, Y, Z\}^{\otimes n}$ appearing in $H'$, we can always find corresponding Pauli operators $P_{0,i}$ and $P_{1,i}$ such that
\begin{eqnarray}
U_B^{\otimes n}(P_{0,i} \otimes P_{1,i}^*)U_B^{\dag \otimes n} = -I_n \otimes Q_i',
\end{eqnarray}
which leads to the identity
\begin{eqnarray}
&&\left(\langle 0|^{\otimes n} \otimes I_n\right) 
\exp\left( \sum_{i=1}^M \lambda_i t\, U_B^{\otimes n}(P_{0,i} \otimes P_{1,i}^*)U_B^{\dag \otimes n} - t \right)
(|0\rangle^{\otimes n} \otimes I_n) \nonumber \\
&=& \left(\langle 0|^{\otimes n} \otimes I_n\right) 
\exp\left( -t\, I \otimes \left( \sum_{i=1}^M \lambda_i Q_i' + I \right) \right) 
(|0\rangle^{\otimes n} \otimes I_n) \nonumber \\
&=& e^{-t(H' + I)}.
\end{eqnarray}
Therefore, by simply setting $t = \beta$, the Lindbladian evolution $e^{\mathcal{L}_H \beta}[\cdot]$ becomes a 1-CBE of the target operator $e^{-\beta(H'+I)}$.

Now, putting everything together, we begin with an $(n+1)$-qubit initial state $\rho_{\text{in}} = |+\rangle\langle +|^{\otimes (n+1)},$
which serves as a $1/2$-NDME of the pure state $|+\rangle^{\otimes n}$. 
Following Eq.~\ref{oorr}, we sequentially apply the quantum channels $\mathcal{C}_{u_2}[\cdot]$, $e^{\mathcal{L}_H \beta}[\cdot]$, and $\mathcal{C}_{u_1^\dag}[\cdot]$, resulting in the output state
\begin{eqnarray}
\rho_{\text{out}} = \mathcal{C}_{u_1^\dag}
\circ 
e^{\mathcal{L}_H \beta} \circ \mathcal{C}_{u_2}[\rho_{\text{in}}].
\end{eqnarray}
From the previous discussion, it follows that $\rho_{\text{out}}$ is an NDME of the quantum state 
\begin{eqnarray}
\mathcal{H}^{\otimes n} U_1^\dag e^{-\beta(H+I)} U_2 \mathcal{H}^{\otimes n} |+\rangle^{\otimes n}
\end{eqnarray}
with amplitude factor
\begin{eqnarray}
\gamma = 2^{-\frac{n_h}{2} - 1},
\end{eqnarray}
where $n_h = n_{h_1} + n_{h_2}$ denotes the total number of Hadamard gates involved in the constructions of $U_1^\dag$ and $U_2$.

To implement the algorithm on a quantum computer and achieve Heisenberg-limited estimation for GCA, we must prepare a purification state $|\psi_{\text{out}}\rangle$ instead of directly preparing the mixed state $\rho_{\text{out}}$. 
For the initial state $\rho_{\text{in}}$, its purification is simply $|\psi_{\text{in}}\rangle = |0\rangle_e \otimes |+\rangle^{\otimes (n+1)},$
where the register $e$ denotes the environment system.
For the channels $\mathcal{C}_{u_2}[\cdot]$ and $\mathcal{C}_{u_1^\dag}[\cdot]$, their corresponding purification unitaries, denoted by $V_{u_2}$ and $V_{u_1^\dag}$, can be constructed by replacing each gate's CBE with its purification unitary, as discussed in Section~\ref{gatecbe}. Therefore, the total circuit depth of the combined unitary $V_{u_1^\dag} V_{u_2}$ is $D = D_1 + D_2$.
Regarding the Lindbladian channel $e^{\mathcal{L}_H \beta}[\cdot]$, we note that its jump operators take the specific form given in Eq.~\ref{special} with $R = 2$, which allows for exponential fast-forwarding. The algorithm described in Section~\ref{provethe2} satisfies our requirements, as it constructs a purification of Lindbladian dynamics. The circuit depth for simulating this evolution is characterized by $\log(K)$, where $K$ denotes the Taylor expansion truncation order, as detailed in Section~\ref{provethe2}.

To estimate the GCA, following Lemma \ref{amp3}, we introduce the following two simple unitaries:
\begin{eqnarray}
U_{r,I} = \begin{pmatrix}
\frac{I+X}{2} \otimes I_n & \cdot \\
\cdot & \cdot
\end{pmatrix}, \quad
U_{i,I} = \begin{pmatrix}
\frac{I+Y}{2} \otimes I_n & \cdot \\
\cdot & \cdot
\end{pmatrix},
\end{eqnarray}
with which we have
\begin{eqnarray}
2 \left|\langle 0| \langle \psi_{\mathrm{out}}| U_{r,I} \otimes I_e |0\rangle |\psi_{\mathrm{out}}\rangle \right| 
&=& 1 + 2^{\frac{n - n_h}{2}} \re\left[\langle \psi_1 | e^{-\beta(H+I)} | \psi_2 \rangle \right], \label{gcareal} \\
2 \left|\langle 0| \langle \psi_{\mathrm{out}}| U_{i,I} \otimes I_e |0\rangle |\psi_{\mathrm{out}}\rangle \right| 
&=& 1 + 2^{\frac{n - n_h}{2}} \im\left[\langle \psi_1 | e^{-\beta(H+I)} | \psi_2 \rangle \right]. \label{gcaimag}
\end{eqnarray}

Then, according to Lemma \ref{thea}, to estimate the real and imaginary parts of the GCA up to an additive error $\epsilon$ with success probability at least $1 - \delta$, the query complexity to the preparation circuit of $|\psi_{\mathrm{out}}\rangle$ is 
\begin{equation}
\mathcal{O}\left( 2^{-\frac{n - n_h}{2}} \epsilon^{-1} \log\left(\delta^{-1}\right) \right).
\end{equation}
Since these queries must be performed adaptively, the overall circuit depth scales as the product of the preparation circuit depth of $|\psi_{\mathrm{out}}\rangle$ and the query complexity.

As discussed previously, the preparation circuit depth of $|\psi_{\mathrm{out}}\rangle$ is $D + \log(K)$. Based on Eqs.~\ref{amp1}--\ref{amp2} and their relation to Eqs.~\ref{gcareal}--\ref{gcaimag}, to ensure the estimation error is within additive error $\epsilon$, we require
\begin{eqnarray}
\sup_{P \in \{I,X,Y,Z\}^{\otimes (n+1)}} \left| \Tr\left( (\tilde{\rho}(t) - \rho(t)) P \right) \right| = \mathcal{O}\left( 2^{\frac{n - n_h}{2}} \epsilon \right).
\end{eqnarray}
Since
\begin{equation}
\left| \Tr\left( (\tilde{\rho}(t) - \rho(t)) P \right) \right| \leq \| \tilde{\rho}(t) - \rho(t) \|_1,
\end{equation}
by the definition of the diamond norm (Definition \ref{diamond}), we can set the simulation error $\varepsilon$ in Theorem \ref{mainthe2} as $\mathcal{O}( 2^{\frac{n - n_h}{2}} \epsilon)$
to meet this requirement. This leads to
\begin{eqnarray}
K = \mathcal{O}\left( \beta + \frac{\log(2^{\frac{ n_h-n}{2}}\epsilon^{-1})}{\log \log(2^{\frac{ n_h-n}{2}}\epsilon^{-1})} \right).
\end{eqnarray}

Because the logarithmic dependence on $\epsilon^{-1}$ in $\log(K)$ is dominated by the factor $\epsilon^{-1}$ in Lemma \ref{thea}, we may approximate $\log(K)$ simply by $\log(\beta)$. Thus, the final circuit depth scales as
\begin{equation}
\mathcal{O}\left( 2^{-\frac{n - n_h}{2}} \epsilon^{-1} \left( M\log(\beta) + D \right) \log\left(\delta^{-1}\right) \right).
\end{equation}

\section{Estimating GCA through QSVT}
\label{qsvt}

Given the Pauli decomposition of $H$
\begin{eqnarray}
H=\sum_{i=1}^M\lambda_i Q_i,
\end{eqnarray}
with $\sum_i \lambda_i = 1$, we can use the LCU technique to construct a unitary of the form that encodes $H$ on the top-left corner, $U_H=\begin{pmatrix}
H & \cdot\\
\cdot & \cdot
\end{pmatrix}.$ It is called a block-encoding of $H$.
To apply LCU, we need the following two unitaries:
\begin{eqnarray}
U_{select} &=& \sum_{i=1}^M |i-1\rangle\langle i-1|\otimes Q_i,\\
U_{\lambda}|0^{\log M}\rangle &=& \sum_{i=1}^M\sqrt{\lambda_i}|i-1\rangle,
\end{eqnarray}
whose circuit depth is $\mathcal{O}(M)$. Then, from LCU, we have
\begin{eqnarray}
U_H =  (U_{\lambda}^\dag\otimes I)U_{select}(U_{\lambda}\otimes I).
\end{eqnarray}

Next, we use QSVT to construct $U_\beta=\begin{pmatrix}
e^{-\beta(H+I)} & \cdot\\
\cdot & \cdot
\end{pmatrix}$
by querying $U_H$ for $\mathcal{O}(\sqrt{\beta}\log(\epsilon^{-1}))$ times, which is also the circuit depth, with inaccuracy $\epsilon$ \cite[Corollary 64]{gilyen2019quantum} (here, we can simply set the approximation error $\varepsilon$ to be the estimation error $\epsilon$). The block-encoding $U_\beta$ encodes the information of the Gibbs-conjugate amplitude (GCA) in the following way: for any two states $\ket{\psi_1}=U_1\ket{0^n}$ and $\ket{\psi_2}=U_2\ket{0^n}$,
\begin{eqnarray}
\langle0^{\log M}|\langle\psi_1|U_\beta|0^{\log M}\rangle|\psi_2\rangle=\langle \psi_1|e^{-\beta (H+I)}|\psi_2\rangle.
\end{eqnarray}

By Lemma \ref{oe}, it suffices to query $U_\beta$, $U_1$, and $U_2$ for $\mathcal{O}(\epsilon^{-1}\log(\delta^{-1}))$ times adaptively to estimate the real and imaginary parts of the GCA up to additive error $\epsilon$. As a result, the overall circuit depth of the QSVT approach are
\begin{eqnarray}\label{unknown}
\widetilde{\mathcal{O}}(\epsilon^{-1}\log(\delta^{-1})(M\sqrt{\beta}+D)),
\end{eqnarray}
where $\widetilde{\mathcal{O}}$ hides polylogarithmic factors ($\log(\epsilon^{-1})$).

Since the spectral norm of any unitary operator is 1, when $\|H\|$ is strictly smaller than 1 and known, we can use the spectral amplification technique \cite{low2017hamiltonian} to construct an amplified block-encoding
\begin{equation}
U_H'=\begin{pmatrix}
\frac{1-\alpha}{\|H\|}H & \cdot\\
\cdot & \cdot
\end{pmatrix}
\end{equation}
with small $\alpha$, by querying $U_H$ adaptively for $\mathcal{O}(\alpha^{-1}\|H\|^{-1}\log(\epsilon^{-1}))$ times. After spectral amplification, we can similarly construct
\begin{eqnarray}
U_\beta'=\begin{pmatrix}
e^{-\beta'(\frac{1-\alpha}{\|H\|}H+I)} & \cdot\\
\cdot & \cdot
\end{pmatrix}
=\begin{pmatrix}
e^{-\beta(H+\frac{\|H\|}{1-\alpha}I)} & \cdot\\
\cdot & \cdot
\end{pmatrix},
\end{eqnarray}
with $\beta'=\frac{\|H\|}{1-\alpha}\beta$, by querying $U_H'$ for $\mathcal{O}(\sqrt{\frac{\|H\|}{1-\alpha}\beta}\log(\epsilon^{-1}))$ times adaptively.

With $U_\beta'$, the GCA is encoded into the following amplitude:
\begin{eqnarray}
\langle0^{\log M}|\langle\psi_1|U_\beta'|0^{\log M}\rangle|\psi_2\rangle=e^{\beta\frac{1-\alpha-\|H\|}{1-\alpha}}\langle \psi_1|e^{-\beta (H+I)}|\psi_2\rangle.
\end{eqnarray}
To estimate the GCA up to additive error $\epsilon$, it suffices to estimate the left-hand side up to additive error $e^{\beta\frac{1-\alpha-\|H\|}{1-\alpha}}\epsilon$. Therefore, the quantum algorithm has query complexity in terms of $U_\beta'$, $U_1$, and $U_2$ of
\begin{equation}
    \mathcal{O}(e^{\beta(\frac{\|H\|}{1-\alpha}-1)}\epsilon^{-1}\log(\delta^{-1})).
\end{equation}
As a result, the total circuit depth are
\begin{eqnarray} \label{eq:d9}
&&\tilde{\mathcal{O}}\left(e^{\beta(\frac{\|H\|}{1-\alpha}-1)}\epsilon^{-1}\log(\delta^{-1})\left(D+M\alpha^{-1}\|H\|^{-1}\sqrt{\frac{\|H\|}{1-\alpha}\beta} \, \right)\right) \nonumber\\
&\approx&\tilde{ \mathcal{O}}\left(e^{\beta(\frac{\|H\|}{1-\alpha}-1)}\epsilon^{-1}\log(\delta^{-1})(D+M\alpha^{-1}\sqrt{\beta})\right),
\end{eqnarray}
which achieves exponentially better scaling in $\beta$ compared with the case in Eq.~\ref{unknown}, where the spectral norm of $H$ is unknown.

\section{Ground state overlap testing\label{gsot}}
As a concrete application of estimating GCA, we consider the task of ground state overlap testing. The idea is that, from the definition of GCA, which is $\langle \psi_1|e^{-\beta (H+I)}|\psi_2\rangle$, we can interpret $e^{-\beta (H+I)}$ as an imaginary time evolution operator. Therefore, as long as there is a non-zero overlap between the ground state $|G\rangle$ of $H$ and $|\psi_2\rangle$, the state $e^{-\beta (H+I)}|\psi_2\rangle$ is essentially proportional to $|G\rangle$ for large enough $\beta$. As a result, GCA can be viewed as the overlap between $|\psi_1\rangle$ and $|G\rangle$. Therefore, our GCA estimation algorithm can be used to probe ground-state properties: how close the ground state $|G\rangle$ is to a given reference state $|\psi_1\rangle$. 

While our algorithm in Theorem \ref{mainthe3} works for arbitrary states $|\psi_1\rangle$ and $|\psi_2\rangle$, for simplicity of analysis and to better demonstrate the advantage of our methods over standard approaches (to be reviewed later), we consider the special case where $|\psi_2\rangle=|+\rangle^{\otimes n}$ and $|\psi_1\rangle=|0\rangle^{\otimes n}$. As for the generalization to arbitrary $|\psi_1\rangle$ and $|\psi_2\rangle$, the derivation proceeds analogously to that in Appendix \ref{proofthe3}, and we therefore omit it here.

We now introduce the basic setup. We have $H=\sum_i E_i|h_i\rangle\langle h_i|$ where $E_i$ is the eigenvalue and $|h_i\rangle$ is the corresponding eigenstate. The set $\{E_i\}$ is ordered in ascending order, and we have $E_0$ the ground energy with ground state $|G\rangle=|h_0\rangle$ and $E_1$ the first-excited state with the gap $E_1-E_0=\Delta$. (Here we assume a unique ground state for simplicity, without loss of generality; the results extend naturally to the degenerate case.)

Under the eigenbasis of $H$, we have the decomposition $|\psi_2\rangle=\sum_i c_i |h_i\rangle$. Now we consider the action of $e^{-\beta (H+I)}$. We first consider the normalized overlap between $e^{-\beta (H+I)}|\psi_2\rangle$ and $|G\rangle$~\cite{shang2024polynomial}. Note that we have
\begin{eqnarray} 
e^{-\beta H}|\psi_2\rangle=\sum_i c_i e^{-\beta E_i} |h_i\rangle,
\end{eqnarray}
therefore, the overlap with $|G\rangle$ has the following expression and the bound
\begin{eqnarray} 
\frac{|\langle G|e^{-\beta (H+I)}|\psi_2\rangle|^2}{\|e^{-\beta (H+I)}|\psi_2\rangle\|^2}&&=\frac{|c_0|^2e^{-2\beta E_0}}{\sum_i |c_i|^2 e^{-2\beta E_i}}\geq \frac{|c_0|^2}{|c_0|^2+\sum_{i\neq 0} |c_i|^2 e^{-2\beta \Delta}}\nonumber\\
&&=\frac{|c_0|^2}{|c_0|^2+(1-|c_0|^2) e^{-2\beta \Delta}}\nonumber\\
&&\geq \frac{|c_0|^2}{|c_0|^2+ e^{-2\beta \Delta}}
\end{eqnarray}
To ensure an overlap of at least $1-\zeta$, we can ask
\begin{eqnarray} 
&&\frac{|c_0|^2}{|c_0|^2+ e^{-2\beta \Delta}}\geq 1-\zeta\nonumber\\
\rightarrow&& \frac{1}{1+ |c_0|^{-2}e^{-2\beta \Delta}}\geq 1-\zeta\nonumber\\
\rightarrow&& 1\geq 1-\zeta+|c_0|^{-2}e^{-2\beta \Delta}-\zeta|c_0|^{-2}e^{-2\beta \Delta}\nonumber\\
\rightarrow&&\zeta(1+ |c_0|^{-2}e^{-2\beta \Delta})\geq |c_0|^{-2}e^{-2\beta \Delta}.
\end{eqnarray}
Since $|c_0|^{-2}e^{-2\beta \Delta}\geq 0$, we can further ask
\begin{eqnarray} 
\zeta\geq |c_0|^{-2}e^{-2\beta \Delta},
\end{eqnarray}
which leads to
\begin{eqnarray} \label{betadelta}
\beta=\mathcal{O}\left(\Delta^{-1}\left(\log(\zeta^{-1})+\log(|c_0|^{-1})\right)\right).
\end{eqnarray}

Now, if our goal is to estimate $\langle \psi_1|G\rangle$ within an error of order $\epsilon\ll1$, we can also set $\zeta=\epsilon$. Following the law of error propagation, to fulfill the goal, we can estimate $\langle \psi_1|e^{-\beta (H+I)}|\psi_2\rangle$ up to an error
\begin{eqnarray} 
\|e^{-\beta (H+I)}|\psi_2\rangle\|\epsilon\approx |c_0|e^{-\beta}e^{-E_0\beta}\epsilon.
\end{eqnarray}
Substituting this and Eq. \ref{betadelta} into Theorem \ref{mainthe3}, our algorithm has the complexity (circuit depth)
\begin{eqnarray}\label{gcaground}
&&\mathcal{O}\left(2^{-n/2}|c_0|^{-1}e^{\beta(1+E_0)}\epsilon^{-1}M\log(\beta)\log\left(\delta^{-1}\right)\right)\nonumber\\=&&\mathcal{O}\left(2^{-n/2}|c_0|^{-1}e^{\beta(1+E_0)}\epsilon^{-1}M\log\left(\delta^{-1}\right)\right)\nonumber\\
=&&\mathcal{O}\left(|c_0|^{-1}\epsilon^{-1}\log\left(\delta^{-1}\right)M2^{-n/2}e^{\Delta^{-1}\left(\log(\epsilon^{-1})+\log(|c_0|^{-1})\right)(1+E_0)}\right)\nonumber\\
=&&\mathcal{O}\left(|c_0|^{-1}\epsilon^{-1}\log(\delta^{-1})M2^{-n/2}(\epsilon^{-1}|c_0|^{-1})^{\frac{1+E_0}{\Delta}}\right).
\end{eqnarray}

In comparison, using the state-of-the-art ground state preparation algorithm \cite{lin2020near} with amplitude estimation, for the same ground overlap testing task, the complexity (circuit depth) is \begin{eqnarray}\label{sotaground}
\mathcal{O}\left(|c_0|^{-1}\epsilon^{-1}\log(\delta^{-1})M\Delta^{-1}\right),
\end{eqnarray}
where we assumed that the block-encoding of $H$ is constructed by LCU.

We can now compare the complexity results between our method Eq. \ref{gcaground} and the previous standard method Eq. \ref{sotaground}. Basically, we are comparing the following two values
\begin{eqnarray}
C_1&&=2^{-n/2}(\epsilon^{-1}|c_0|^{-1})^{\frac{1+E_0}{\Delta}},\nonumber\\
C_2&&=\Delta^{-1},
\end{eqnarray}
and we are interested in the regimes where $C_1\ll C_2$ so that our methods have advantages. 

The most obvious case in which our method achieves an exponential advantage is when $E_0=-1$. According to the definition of the Hamiltonian Eq. \ref{mainhami}, which satisfies $\sum_{i=1}^M\lambda_i=1$, the condition $E_0=-1$ means that $H$ is frustration-free~\cite{michalakis2013stability}; that is, $|G\rangle$ is the ground state for any Pauli terms in $H$. In this case, the comparison becomes $C_1=2^{-n/2}\ll1$ versus $C_2=\Delta^{-1}\geq 1$, where the advantage is significant. However, $E_0=-1$ basically means $H$ is a stabilizer Hamiltonian, which is classically easy. In general, we should expect $1+E_0>0$. Therefore, $C_1$ is determined by the relationship between $1+E_0$, which serves as a measure of frustration to some extent, and $\Delta$, the spectral gap of the Hamiltonian.

We can now consider a more realistic and classically intractable scenario. Suppose $|\psi_2\rangle$ is a guided initial state (i.e., $|c_0|^{-1}=\mathcal{O}(\text{poly} (n))$), then preparing the ground state becomes the guided local Hamiltonian problem, which is BQP-complete \cite{cade2022improved}. Recall that the overlap is upper bounded by $2^{-n/2}$; to have a meaningful estimation, we need to require $\epsilon=\mathcal{O}(2^{-n/2})$. Putting these into $C_1$, we have
\begin{eqnarray}
\mathcal{O}(C_1)=\mathcal{O}\left(\left(2^{-n/2}\right)^{1-\frac{1+E_0}{\Delta}}\right).
\end{eqnarray}
Therefore, as long as the frustration $1+E_0$ is smaller than the Hamiltonian gap $\Delta$, our method would still have exponential advantages over the standard method Eq. \ref{sotaground}.

When $|\psi_2\rangle$ is not a guided state, then on average, when $|\psi_2\rangle$ is sampled from 1-design~(\cite{mele2024introduction}, example 52), the average $c_0$ would be $2^{-n/2}$. Putting this into $C_1$, we have
\begin{eqnarray}
\mathcal{O}(C_1)=\mathcal{O}\left(\left(2^{-n/2}\right)^{1-2\frac{1+E_0}{\Delta}}\right).
\end{eqnarray}
Similarly, as long as the frustration $2(1+E_0)$ is smaller than the Hamiltonian gap $\Delta$, our method would still have exponential advantages over the standard method Eq. \ref{sotaground}.

\end{appendix}
\end{document}